\documentclass{jams-l}
\usepackage[left=1in,right=1in,bottom=1in,top=1in]{geometry}
\usepackage{amssymb}
\usepackage{graphicx}
\usepackage{tikz-cd}
\usepackage{undertilde}
\usepackage{enumitem}
\usepackage{mathtools}
\usepackage{hyperref}
\hypersetup{colorlinks=true}
\copyrightinfo{2021}{Nivedita}
\newtheorem{theorem}{Theorem}[section]
\newtheorem{lemma}[theorem]{Lemma}
\theoremstyle{definition}
\newtheorem{definition}[theorem]{Definition}
\theoremstyle{definition}
\newtheorem{assumption}[theorem]{Assumption}
\newtheorem{example}[theorem]{Example}

\theoremstyle{remark}
\newtheorem{remark}[theorem]{Remark}
\numberwithin{equation}{section}

\begin{document}
    \title[Topological classification of crystal defects]{Revisiting the topological classification of defects in crystals}
    \author[Nivedita and A. Gupta]{Nivedita | Anurag Gupta}
    \address{IIT Kanpur, Kanpur, UP-208016, India}
    \email{nivdita@iitk.ac.in, ag@iitk.ac.in}
   
    \date{\today}
    
    \begin{abstract}
        A general theory of topological classification of defects is introduced. We illustrate the application of tools from algebraic topology, including homotopy and cohomology groups, to classify defects including several explicit calculations for crystals in $\mathbb{R}^2$, $S^2$, 2-dimensional cylinder, 2-dimensional annulus, and 2-tori. A set of physically motivated assumptions is formulated in order to justify the classification process and also to expose certain inherent inconsistencies in the considered methodology, particularly for crystal lattices.
    \end{abstract}
    \maketitle
    \tableofcontents
    \section{Introduction}
     Defects and textures in crystalline materials yield strong influence over a wide range of their physical properties, from mechanical strength to electronic and magnetic behaviour. An important concern with defects in crystals is to be able to classify them, and understand their equivalence, depending on the physical space they may inhabit. In particular, the classification results will vary between crystals in Euclidean spaces and those on curved spaces (e.g., spheres, cylinders, tori). With the idea of defects being related to breaking of order and symmetry, it is natural to look at them through topological and group theoretic methods \cite{Mermin79, Kleman89}.  
    

   The Order parameter space is a space of all possible ground states of a system (or vacua). A defect represents a region in the physical space where the order parameter no longer takes values in the order parameter space. On the other hand, a given property is called topological if continuous deformations do not lead to a change in that property. Given a configuration, the order parameter is a function from the real space to the order parameter space. If this map is undefined at some places, and can not be extended continuously to the undefined regions, the defect is called topological. If it is possible to extend the function over a defect continuously then the defect is topologically unstable or topologically trivial. Of central interest to us, therefore, is to deal with continuous deformations of functions between topological spaces (the real space and the order parameter space). This is achieved through homotopy theory, which is the study of mathematical contexts where the functions are equipped with a concept of deformations (morphism of the `second level'), and then iteratively with homotopies of homotopies between those (morphisms of the `third level').

One of the first instances of using homotopy theoretic ideas for kinks (textures) in a field-theoretic setting was by Finkelstein~\cite{Finkelstein66} who also discussed, albeit briefly, the special case of continuous crystals. Rogula~\cite{Rogula} extended this work to demonstrate how one can use homotopy theory to classify localised defects (such as point or line defects in $\mathbb R^3$). However, he took the symmetry group of crystals to be the affine group which, being inconsistent with the physics of crystals, led to erroneous results such as the non-existence of domain walls \cite{Michel78}.
A systemic classification of defects of various dimensionality using homotopy theory was given by Toulouse and Kl\`eman \cite{GK76}, who introduced the idea of surrounding the defects with spheres and classifying the possible defects via the fundamental group (and higher homotopy groups), cf. \cite{Mermin79, Balian1981PhysiqueDD}. The classification methodology was applied to crystals, both Euclidean and curved, by Kl\'eman \cite{Kleman89}, however falling short of explicitly calculating the equivalence classes of defects.
    
The purpose of this work is to revisit the theory of topological classification of defects in the context of crystalline materials. Towards this end, we derive a manifold generalisation of the classification expression used otherwise in the literature. The generalisation allows us to use both homotopy groups (fundamental group and higher) and cohomology groups as computational tools for the classification purposes. The latter is useful, in particular, for crystals in general manifolds (non-Euclidean and non-spherical), as demonstrated explicitly for the cases of 2-dimensional (2D) cylinder, 2D annulus, and 2D torus. The homotopy methods, on the other hand, are useful for crystals in Euclidean and spherical spaces, as illustrated in the accompanying calculations for various lattice structures in $\mathbb{R}^2$ and $S^2$. Most of our results, obtained using homotopy and cohomology groups, have not appeared in the earlier literature, to the best of our knowledge.

In the second part of this work we formulate a set of physically meaningful assumptions based on which the classification theory can be justified. The precise assumptions are given in terms of the topology of the system's configuration space. Having these well-stated assumptions in hand, we are led to highlight a few inconsistencies that may limit the validity and applicability of the classification results for crystalline materials. We also propose some directions in order to partially rectify the inconsistencies. More specifically, we give a definition of disinclination for general manifolds and introduce the idea of using a sheaf theoretic formalism for defects.

We provide a brief outline of the article. In Section~\ref{sec:introchaptopdef}, we derive a general mathematical expression for defect classification, following with a discussion on homotopy and cohomology groups in Section~\ref{sec:homcoho}. Beginning with Section~\ref{sec:crystals} we restrict our attention to crystals. The topological classification of defects for crystals in $\mathbb{R}^n$ and in $S^n$, with several detailed examples, is provided in Sections~\ref{sec:crystalsE} and \ref{sec:crystalsS}, respectively. In Section~\ref{sec:crystalsC}, we use cohomology groups to classify defects on crystals on more general manifolds. A set of well motivated assumptions, which justify the classification process of the preceding sections, is given in Section~\ref{sec:phyjust}. Building upon the assumptions, in Section~\ref{sec:critique}, we discuss several inconsistencies that arise naturally in the classification process. In Section~\ref{reconcile}, we close the article by suggesting alternatives in order to rectify some of the inconsistencies. 
    
\section{\centering Formulation} 
\label{sec:introchaptopdef}
A physical theory is defined by specifying a physical space $M$, with some kind of geometric structure (e.g., metric tensor or the structure of a principal bundle), a configuration space $\mathcal C_M$ of allowable physical configurations on $M$, and a smooth Hamiltonian map $\mathcal H_M\colon\mathcal C_M\to\mathbb R$ which assigns energy to each configuration.
\begin{definition}
    Let $V_M\subseteq\mathcal C_M$ be the set of local minima of $\mathcal H_M$ and $G_M\subseteq\mathrm{Aut}\left(\mathcal C_M\right)$ be the (geometric) automorphisms of $\mathcal C_M$ that preserve $\mathcal H_M$, i.e.,
    \begin{equation}
        \begin{split}
            V_M=&\left\{\phi\in C_M\mid\phi\text{ is a minimum of }\mathcal H_M\right\}\\
            G_M=&\left\{g\in\mathrm{Aut}\left(\mathcal C_M\right)\mid\mathcal H_M\circ g=\mathcal H_M\right\}.
        \end{split}
    \end{equation}
In physical terms we can associate $G_M$ with the symmetry group of the system and $V_M$ with the space of ground-states or vacua.
\end{definition}

\begin{remark}
    \label{rem:notation}
    Hereafter we will not write the subscript $M$ explicitly whenever the dependence is clear from the context.
\end{remark}
In a physical theory with degenerate vacua, where $V$ is not a singleton, a defect, defective on a subset $X$ of the space $M$, is an assignment of a vacuum to each point of $M\setminus X$ which can not be extended to a global function from $M$ even up to a homotopy. The topological classification of defects provides all the equivalence classes of possible defects structures, defective on $X$, up to a homotopy. There is no continuous path (deformation) between two non-equivalent defect structures. Topological defects are said to be stable if they are not homotopic to the constant function which assigns the same vacuum to every point. Generalising \cite{Mermin79}, we introduce
\begin{definition}
    \label{def:def}
    For $X\subseteq M$, we define the set of topological defects, defective on $X$, as
    \begin{equation}
        \mathrm{Def}_M\left(X\right)=\mathrm{hTop}\left[M\setminus X,V\right].
    \end{equation}
\end{definition}
\noindent We will physically motivate this definition in Section~\ref{sec:phyjust}.

\begin{remark}
    We look at homotopy classes of continuous functions instead of smooth functions because smooth functions are dense in continuous functions between manifolds. Moreover, if two smooth functions are continuously homotopic then they are smoothly homotopic~\cite[Theorem 8]{Pontryagin55}. Therefore, the homotopy type calculated using either smooth or continuous functions is the same~\cite[Propostion 17.8, Corollary 17.8.1]{BT82}. 
\end{remark}
\begin{remark}
    \label{rem:bdry}
    If $M$ is not compact, or if $M$ has a boundary, we may need boundary conditions to ensure that the Hamiltonian is well-defined or behaves well under (functional) variations. Whenever the boundary conditions are analytic, we can consider the subspace of smooth functions satisfying the boundary conditions. On the other hand, if we need functions that are constant at infinity then we can replace $M$ by its $1$-point compactification in Definition~\ref{def:def}. The $1$-point compactification adds a point at infinity. That a function is constant at infinity is then same as assigning that constant to the point at infinity.
\end{remark}
To make use of Definition~\ref{def:def}, we need to understand the homotopy-theoretic structure of $V$. The action of $G$ on any element of $V$ maps it to another element of $V$ as $G$ preserves $\mathcal H$ and hence its minima. This gives us an action of $G$ on $V$. We assume that, for $v\in V$, the natural map $G/G_v\to Gv$ is a homeomorphism, where $G_v$ is the stabiliser of $v$ and $Gv$ is its orbit. The sufficient conditions for this map to be a homeomorphism are given in Theorems~\ref{theo:properorbstab} and \ref{theo:lchcorbstab}. We further assume that $V$ is the (topological) coproduct of the orbits. Together, these two assumptions imply
\begin{equation}
    \label{eq:decompminima}
    V\cong\coprod_{v\in V_a}G/G_v,
\end{equation}
where $V_a$ is a choice of configuration from each orbit. Physically, it is a choice of ground-state from each nonequivalent type of minimum.

In the following theorem we obtain a general mathematical expression for $\mathrm{Def}_M\left(X\right)$, as introduced in Definition~\ref{def:def}, taking into account the simplifications discussed in the previous paragraph. 
\begin{theorem}[Classification]
    \label{theo:defclass}
    If $V$ satisfies Equation~\ref{eq:decompminima}, and both $M\setminus X$ and $V$ are locally path-connected, then
    \begin{equation}
    \label{eq:expandedexpression}
        \mathrm{Def}_M\left(X\right)\cong\prod_{A\in\pi_0\left(M\setminus X\right)}\coprod_{v\in V_a}\mathrm{hTop}\left[A,G_0/\left(G_0\cap G_v\right)\right]\times\pi_0\left(G\right)/p\left(G_v\right),
    \end{equation}
    where $p\colon G\to\pi_0\left(G\right)$ is the quotient map.
\end{theorem}
\noindent According to this result a defect, defective on $X\subseteq M$, is an element of $\mathrm{Def}_M\left(X\right)$ given by
\begin{enumerate}
    \item A choice of an element of $V_a$ for every connected component of the non-defective region. In crystals this is a choice of the crystalline structure in each domain.
        \item An element of $\pi_0\left(G\right)/p\left(G_v\right)$. Physically, nontrivial elements $\pi_0\left(G\right)$ are the symmetries which cannot be deformed to identity (like reflections of crystals) and $p\left(G_v\right)$ are symmetries that preserve the vacuum $v$. For crystals, this is a choice of chirality in different domains. This generalisation allows for inclusion of twin-boundary like defects.
        \item An element of $\mathrm{hTop}\left[A,G_0/\left(G_0\cap G_v\right)\right]$. This is a choice of assigning each point in the domain to a vacuum symmetric to the chosen vacuum type $v$. In crystals, this corresponds to an assignment of a rotated and translated lattice (and thus unit cell) to each point in the domain.
\end{enumerate}
\begin{proof}[Proof of Theorem~\ref{theo:defclass}]
    Applying Theorem~\ref{theo:breakintopathcomp} to Definition~\ref{def:def} we obtain
    \begin{equation}
        \mathrm{Def}_M\left(X\right)\cong\prod_{A\in\pi_0\left(M\setminus X\right)}\coprod_{B\in\pi_0\left(V\right)}\mathrm{hTop}\left[A,B\right].
    \end{equation}
    On the other hand, application $\pi_0$ to Equation~\ref{eq:decompminima}, and using Theorem~\ref{theo:pi0prescoprod}, yields
    \begin{equation}
        \pi_0\left(V\right)\cong\coprod_{v\in V_a}\pi_0\left(G/G_v\right).
    \end{equation}
    Putting these together, we obtain
    \begin{equation}
        \mathrm{Def}_M\left(X\right)\cong\prod_{A\in\pi_0\left(M\setminus X\right)}\coprod_{v\in V_a}\coprod_{B\in\pi_0\left(G/G_v\right)}\mathrm{hTop}\left[A,B\right].
    \end{equation}
    Note that, since $G/G_v$ is a homogeneous $G$-space, all the connected components are homeomorphic (Theorem~\ref{theo:connhomhomeo}). Any $B\in\pi_0\left(G/G_v\right)$ is therefore homeomorphic to $\left(G/G_v\right)_0$ (the component containing the coset $G_v$). Consequently, we can write
    \begin{equation}
        \mathrm{Def}_M\left(X\right)\cong\prod_{A\in\pi_0\left(M\setminus X\right)}\coprod_{v\in V_a}\mathrm{hTop}\left[A,\left(G/G_v\right)_0\right]\times\pi_0\left(G/G_v\right).
    \end{equation}
    The desired result follows on simplifying this expression using Theorems~\ref{theo:idquot} and \ref{theo:connquot}.
\end{proof}
\begin{remark}
  Nondefective configurations are represented by those defects in which each domain is given the same choice of vacuum, and further, each domain is given a constant function, the constant being the same for all domains. The element of $\pi_0\left(G\right)/p\left(G_v\right)$ is also same in each domain.
\end{remark}
\noindent In the next section, we discuss two standard techniques from algebraic topology which will allow us to calculate the homotopy classes of maps between two given spaces.

\section{\centering Relation to Homotopy Groups and Cohomology}
\label{sec:homcoho}
For the classification we need to calculate $\mathrm{hTop}\left[A,G_0/\left(G_0\cap G_v\right)\right]$. Depending on the topological nature of $A$ and $G_0/\left(G_0\cap G_v\right)$, we may in some special cases be able to calculate $\mathrm{hTop}\left[A,G_0/\left(G_0\cap G_v\right)\right]$ either in terms of homotopy groups of $G_0/\left(G_0\cap G_v\right)$ or in terms of cohomology groups of $A$. For instance, if the domain can be expressed as a pointed sums of spheres then we use homotopy groups. However, when the codomain is expressible as products of $S^1$ (more generally, products of spaces that are terms of some spectra) we use cohomology groups. 
\subsection{\centering Homotopy groups}
Consider the following situations:
\begin{itemize}
    \item $M=\mathbb R^n$ and $X$ is a union of $k_m$ linear subspaces of dimension $m$ which do not intersect (for $m=0,\dots,n-1$). Then $M\setminus X$ is homotopy equivalent to
    \begin{equation}
        \coprod_{i=0}^{k_{n-1}}\bigvee_{j=0}^{n-2}\left(S^{n-j-1}\right)^{\vee k_{ij}},
    \end{equation}
    where $k_{ij}$ is the number of $j$-dimensional subspaces lying between the $n-1$-dimensional subspaces numbered $i$ and $i+1$ (the $0^\mathrm{th}$ and the $\left(k_{n-1}+1\right)^\mathrm{th}$ subspace are both at infinity). This is shown in Figure~\ref{fig:rnretract} for $n=3$.
    \item $M=\mathbb R^3$ and $X$ is a circle. Then $M\setminus X$ is homotopy equivalent to $S^2\vee S^1$, see  Figure~\ref{fig:rnminuscircle}.
    \item $M=S^n$ and $X$ is a set of $k$-points ($k>0$). Then $M\setminus X$ is homotopy equivalent to $\left(S^{n-1}\right)^{\vee k-1}$, see Figure~\ref{fig:sphereretract} for $n=2$.
\end{itemize}
\begin{figure}
    \centering
    \includegraphics[scale=0.7]{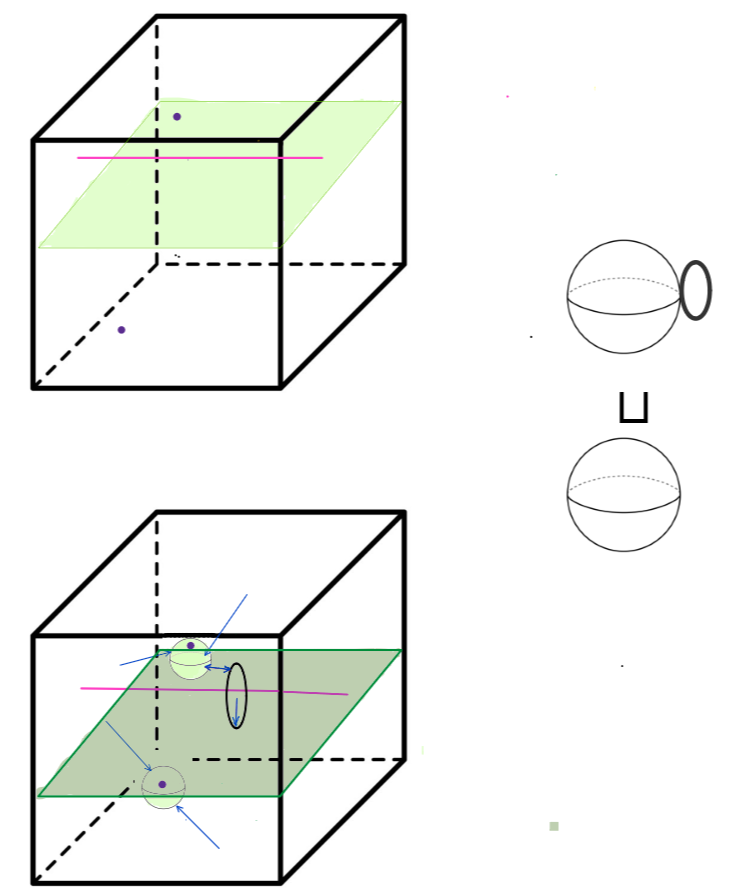}
    \caption{$\mathbb R^3$ minus a plane, a line, and two points. The part of $\mathbb{R}^3$ below the plane retracts to $S^2$ around the singular point. Similarly, above the plane, the retraction (indicated by the arrows) gives us a wedge between $S^1$ and $S^2$.}
    \label{fig:rnretract}
\end{figure}

\begin{figure}
    \centering
    \includegraphics[scale=0.6]{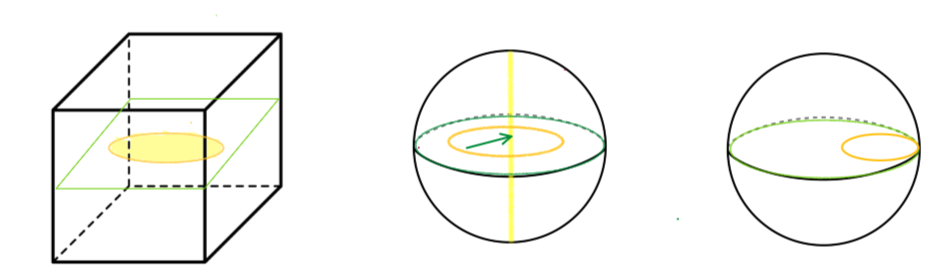}
    \caption{$\mathbb R^3$ minus a circle. The shaded yellow disc contracts around a line segment, while the rest of the space retracts to form a sphere $S^2$. We can slide one end of the diameter along the surface of the sphere to get the final wedge of $S^2$ and $S^1$.}
    \label{fig:rnminuscircle}
\end{figure}

\begin{figure}
    \centering
    \includegraphics[scale=0.4]{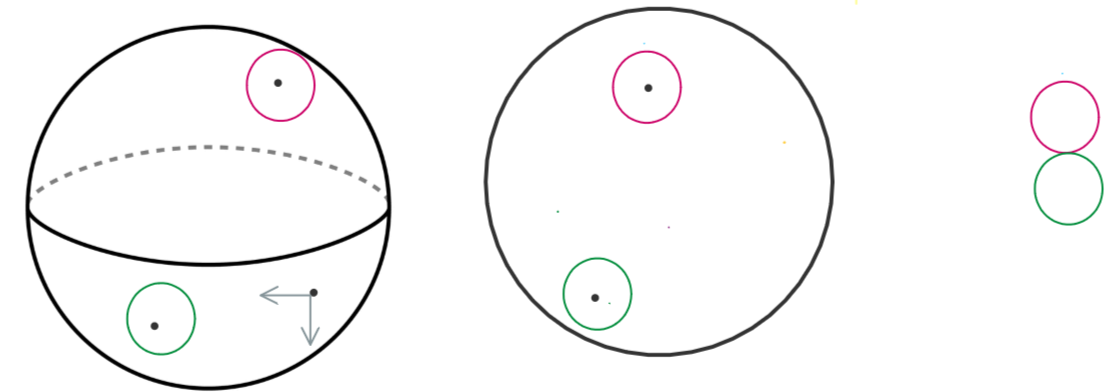}
    \caption{$S^2$ minus points. On removing the first point the sphere becomes homotopic a disc. Further removal of points yields a wedge of circles.}
    \label{fig:sphereretract}
\end{figure}
In several cases, such as the ones exhibited above, the connected components of $M\setminus X$, i.e., $A\in\pi_0\left(M\setminus X\right)$, will often be homotopy equivalent to $\bigvee_{i=1}^kS^{n_i}$. Then
\begin{equation}
    \label{eq:htpygroupsarise}
    \mathrm{hTop}\left[A,G_0/\left(G_0\cap G_v\right)\right]\cong\frac{\mathrm{hTop}_*\left[A,G_0/\left(G_0\cap G_v\right)\right]}{\pi_1\left(G_0/\left(G_0\cap G_v\right)\right)}\cong\prod_{i=1}^k\frac{\pi_{n_i}\left(G_0/\left(G_0\cap G_v\right)\right)}{\pi_1\left(G_0/\left(G_0\cap G_v\right)\right)},
\end{equation}
where the first step follows from Theorem~\ref{theo:relbtwbasdunbshtpy} and the second from Theorem~\ref{theo:wedgeiscoprod} and Theorem~\ref{theo;pinquotfreehom}. The quotient is by the standard action of the fundamental group on pointed mapping spaces. The following lemma collects the conditions under which the class of defects have a natural group structure.
\begin{lemma}
    If 
    \begin{itemize}
        \item $V_a$ is a singleton (let $v$ denote the unique element),
        \item $M$ and $X\subseteq M$ are such that each $A\in\pi_0\left(M\setminus X\right)$ is homotopy equivalent to a wedge of spheres, i.e.,
        \begin{equation*}
            A\simeq\bigvee_{i=1}^{k\left(A\right)}S^{n_i\left(A\right)},
        \end{equation*}
        \item For every $n_i\left(A\right)$, the action of $\pi_1\left(G_0/\left(G_0\cap G_v\right)\right)$ on $\pi_{n_i\left(A\right)}\left(G_0/\left(G_0\cap G_v\right)\right)$ is trivial, and
        \item$p\left(G_v\right)\trianglelefteq\pi_0\left(G\right)$,
    \end{itemize}
    then $\mathrm{Def}_M\left(X\right)$ has a natural group structure.
\end{lemma}
\noindent In case of crystals, this group structure is related to the interaction and entanglement of line defects \cite[Chapter 9]{Sethna21}.
\begin{proof}
    Under the present considerations, Equation~\ref{eq:htpygroupsarise} reduces the classification to
    \begin{equation}
        \begin{split}
            \mathrm{Def}_M\left(X\right)&\cong\prod_{A\in\pi_0\left(M\setminus X\right)}\left(\left(\prod_{i=1}^{k\left(A\right)}\pi_{n_i\left(A\right)}\left(G_0/\left(G_0\cap G_v\right)\right)\right)\times\pi_0\left(G\right)/p\left(G_v\right)\right)\\
            &\cong\pi_0\left(G\right)/p\left(G_v\right)^{\pi_0\left(M\setminus X\right)}\times\prod_{A\in\pi_0\left(M\setminus X\right)}\prod_{i=1}^{k\left(A\right)}\pi_{n_i\left(A\right)}\left(G_0/\left(G_0\cap G_v\right)\right).
        \end{split}
    \end{equation}
    Since $p\left(G_v\right)\trianglelefteq\pi_0\left(G\right)$, this is a product of groups.
\end{proof}
 
\begin{example}
    \label{ex:wedgeofS1pi}
    Let $X\subseteq M$ be such that $M\setminus X\simeq\left(S^1\right)^{\vee m}$. Then Equation~\ref{eq:htpygroupsarise} gives us
    \begin{equation}
        \label{eq:merminpi1conj}
        \mathrm{Def}_M\left(X\right)\cong\coprod_{v\in V_a}\left(\text{conjugacy classes of }\pi_1\left(G_0/\left(G_0\cap G_v\right)\right)\right)^m\times\pi_0\left(G\right)/p\left(G_v\right),
    \end{equation}
    where we have used the fact that $\pi_1$ acts on itself by conjugation (Theorem~\ref{theo:pi1selfocnj}).  In \cite{Mermin79,Kleman77}, line defects in $\mathbb R^3$ are classified by conjugacy classes of $\pi_1\left(G/G_v\right)$. Since $\mathbb R^3\setminus\text{line}\simeq S^1$, we see that Equation~\ref{eq:merminpi1conj} applies with $m=1$. If the fundamental group happens to be abelian then we have a natural group structure.  Our result contains additional terms as we have considered the case where $G$ is not necessarily connected and have allowed for the possibility of nonequivalent minima. It reduces to the formula in \cite{Mermin79} when these generalisations are removed.

\end{example}
\subsection{\centering Cohomology}
Let $E_\alpha$ be a spectra representing cohomology theories $H_\alpha$ ($\alpha=$ singular, K-theory, stable cohomotopy etc) \cite{nlab:generalized_(eilenberg-steenrod)_cohomology}. If $G_0/\left(G_0\cap G_v\right)\simeq\prod_{i\in I}E_{\alpha_i}\left(n_i\right)$ (where $E\left(n\right)$ is the $n$\textsuperscript{th} space in the spectrum $E$) then
\begin{equation}
    \mathrm{hTop}\left[A,G_0/\left(G_0\cap G_v\right)\right]\cong\mathrm{hTop}\left[A,\prod_{i\in I}E_{\alpha_i}\left(n_i\right)\right]\cong\prod_{i\in I}\mathrm{hTop}\left[A,E_{\alpha_i}\left(n_i\right)\right]\cong\prod_{i\in I}H_{\alpha_i}^{n_i}\left(A\right).
\end{equation}
Thus, in this case, calculating $\mathrm{Def}_M\left(X\right)$ reduces to calculating the generalised cohomology groups of the connected components of $M\setminus X$ and some discrete calculations involving $V_a$ and $\pi_0\left(G\right)$.
\begin{example}
    \label{ex:cohomology}
    If $G_0/\left(G_0\cap G_v\right)\simeq\mathbb T^n$, we obtain
    \begin{equation}
        \mathrm{hTop}\left[A,G_0/\left(G_0\cap G_v\right)\right]\cong H_{\mathrm{sing},\mathbb Z}^1\left(A\right)^n,
    \end{equation}
    as $\mathbb T^n\cong\left(S^1\right)^n$ and $S^1\cong K\left(\mathbb Z,1\right)=E_{\mathrm{sing},\mathbb Z}\left(1\right)$.
   Here, we have used the Eilenberg-Maclance spaces $K\left(G,n\right)$ which are pointed spaces defined by
    \begin{equation}
        \pi_k\left(K\left(G,n\right)\right)\cong
        \begin{cases}
            G\text{ if }n=k,\\
            0\text{ otherwise.}
        \end{cases}
    \end{equation} 
    Clearly, to define $K\left(G,0\right)$ we require $G$ to be a set, to define $K\left(G,1\right)$ we require $G$ to be a group, and to define $K\left(G,n\right)$ we require $G$ to be an abelian group. By definition, this space is unique up to weak homotopy equivalence.  
    If $G$ is an abelian group, then the spaces $K\left(G,n\right)$ form a spectrum $HG$ \cite{nlab:eilenberg-mac_lane_spectrum}. 
    Previously, twisted cohomology has been used in the classification of global textures in nematics~\cite{doi:10.1098/rspa.2016.0265}. 
\end{example}

\section{\centering Applications to Crystals}
    \label{sec:crystals}
    The configuration space of crystals is typically that of particles (finite or infinite) living on a Riemannian manifold with the Hamiltonian depending only on distances between particles. Thus we have \cite{Kleman89}
    \begin{equation}
        G_M=\mathrm{Isom}\left(M\right).
    \end{equation} 
    On quotienting $V$ by isometries, we find that $V_a$ corresponds to various possible crystal structures that the material can exist in. When we study the defects possible in a particular crystal structure, $V_a$ will considered to be a singleton. Then $G_v$ is the collection of isometries that preserve the lattice $v$ and hence can be identified with the symmetry group of the crystal structure that $v$ belongs to. It acts on $M$ and has a fundamental domain as the unit cell of the crystal, which we assume to be compact and of the same dimension as $M$. This forces $G_v\leq\mathrm{Isom}\left(M\right)$ to be a discrete cocompact subgroup. Indeed, the unit cell is diffeomorphic to $M/G_v$ and consequently $\dim M=\dim M/G_v$. This implies $\dim G_v=0$. We will classify defects for discrete cocompact groups of the isometry group, corresponding to different crystal structures.
   
    Let $q_M\colon\widetilde{\mathrm{Isom}\left(M\right)_0}\to\mathrm{Isom}\left(M\right)_0$ be a universal cover (which always exists for any connected manifold). In the case of crystals, $G_v$ is discrete with vanishing positive homotopy groups. Under such a circumstance Theorem~\ref{theo:lesfunda} tells us that
    \begin{equation}
        \label{eq:discretepihigh0}
        \begin{split}
            \pi_n\left(\mathrm{Isom}\left(M\right)_0/\left(\mathrm{Isom}\left(M\right)_0\cap G_v\right)\right)&\cong\pi_n\left(\mathrm{Isom}\left(M\right)_0\right)\text{ for }n>1\\
            &\cong\pi_n\left(\widetilde{\mathrm{Isom}\left(M\right)_0}\right),
        \end{split}
    \end{equation}
    where in the second line we have used the fact that $\pi_n$ of a space and its universal cover are equal for $n>1$. To calculate $\pi_1$, we use Theorem~\ref{theo:univcovpi1} to obtain
    \begin{equation}
        \label{eq:pi1univcovsub}
        \pi_1\left(\mathrm{Isom}\left(M\right)_0/\left(\mathrm{Isom}\left(M\right)_0\cap G_v\right)\right)\cong q^{-1}\left(\mathrm{Isom}\left(M\right)_0\cap G_v\right)\leq\widetilde{\mathrm{Isom}\left(M\right)_0}.
    \end{equation}
    The action of $\pi_1$ on $\pi_n$ is given by restricting the obvious action of $\widetilde{\mathrm{Isom}\left(M\right)_0}$ on $\pi_n\left(\widetilde{\mathrm{Isom}\left(M\right)_0}\right)$.

\section{\centering Crystals in Euclidean spaces}
\label{sec:crystalsE}
For crystals in $\mathbb R^n$ we take, in accordance with Section~\ref{sec:crystals},
\begin{equation}
    G=\mathrm{Isom}\left(\mathbb R^n\right)\cong\mathbb R^n\rtimes\mathrm O\left(n\right),
\end{equation}
where $\mathrm O\left(n\right)$ acts on $\mathbb R^n$ in the natural way. So we have $\pi_0\left(G\right)\cong\mathbb Z/2\mathbb Z$ and $G_0\cong\mathbb R^n\rtimes\mathrm{SO}\left(n\right)$. Firstly, this tells us that $\pi_0\left(G\right)/p\left(G_v\right)$ is either trivial or $\mathbb Z/2\mathbb Z$ depending on whether $G_v$ contains a reflection or not. Secondly, as discussed in Section~\ref{sec:crystals}, we have to look for discrete cocompact subgroups $G_0\cap G_v$ of $\mathbb R^n\rtimes\mathrm{SO}\left(n\right)$. By Bieberbach's theorems~\cite[Theorem 2.1]{Szczepanski12}, we know that $\mathbb A=_\text{def}G_v\cap\left(\mathbb R^n\rtimes\left\{I_n\right\}\right)$ is isomorphic to $\mathbb Z^n$ (this isomorphism is basically a choice of basis for the lattice). Therefore $G_0\cap G_v\cong\mathbb A\rtimes C$, where $C\subseteq\mathrm{SO}\left(n\right)$ are those rotations that preserve the lattice $\mathbb A$. The subspace $C$ will be finite, for an infinite subspace of a compact space $\mathrm{SO}\left(n\right)$ cannot be discrete, and will hence admit a finite set of generators $c_1,\dots,c_x$. Choosing a basis for $\mathbb A$ we can write
\begin{equation}
    A\rtimes C\cong\mathbb Z^n\rtimes_{M_1,\dots,M_x}C,
\end{equation}
where the $M_i$ are the coordinate representations for $c_i$.
\subsection{\centering Geometry of order parameter space}
\label{subsec:geometryRn}
We now look into the nature of the space $\mathbb R^n\rtimes\mathrm{SO}\left(n\right)/\mathbb A\rtimes C$. By the definition of coset space, we identify
    \begin{equation}
        \left(v,R\right)\sim\left(b+gv,gR\right)\text{ for all } R\in\mathrm{SO}\left(n\right), v\in \mathbb R^n, b\in\mathbb A\text{ and }g\in C.
    \end{equation}
    Choosing $g=I_2$, we observe that $\left(v,R\right)\sim\left(b+v,R\right)$. This collapses the quotient to $\left(\mathbb T^n\rtimes\mathrm{SO}\left(n\right)\right)/C$.
    
    \begin{lemma}
        $\left(\mathbb T^n\rtimes\mathrm{SO}\left(n\right)\right)/C$ is a fibre bundle  over $\mathrm{SO}\left(n\right)/C$ with fibre $\mathbb T^n$.
    \end{lemma}
    
    \begin{proof}
        Consider the bundle $\mathrm{pr}_1\colon\mathbb T^n\rtimes\mathrm{SO}\left(n\right)\to\mathrm{SO}\left(n\right)$ (fibre is $\mathbb T^n$). This bundle is trivial as a topological bundle but with nontrivial group structure. The action of $C$ is given by
        \begin{equation}
            g\left(v,R\right)=\left(gv,gR\right),
        \end{equation}
        implying that $g\cdot \mathrm{pr}_1\left(x,R\right)=\mathrm{pr}_1\left(g\cdot\left(x,R\right)\right)$. Since $ C$ acts equivariantly on this trivial bundle, the quotient is a bundle with the same fibre, i.e., we obtain the bundle $\pi\colon\left(\mathbb T^n\rtimes\mathrm{SO}\left(n\right)\right)/C\to\mathrm{SO}\left(n\right)/C$ with fibre $\mathbb T^n$.
    \end{proof}
    \noindent Figure~\ref{fig:torusbundleovercircle} illustrates the fibre bundle structure of the order parameter space for $n=2$.  The space $\left(\mathbb T^2\rtimes\mathrm{SO}\left(2\right)\right)/C$ is illustrated as a $\mathbb T^2$-bundle over $S^1$ since any discrete subgroup $C\leq\mathrm{SO}\left(2\right)\cong S^1$ will be cyclic (i.e., $\mathrm{SO}\left(2\right)/C\cong S^1$).
    
    \begin{remark}[Unbroken rotational symmetry]
        If we wish to consider defects only due to breaking of translation symmetry we take $C=\mathrm{SO}\left(n\right)$ (this is not discrete, so it is technically not a crystal), then the space $\left(\mathbb T^n\rtimes\mathrm{SO}\left(n\right)\right)/C$ reduces to $\mathbb T^n$ as it is a $\mathbb T^n$-bundle on a single point.
    \end{remark}
    \begin{remark}[Completely broken rotational symmetry]
        In contrast to the above example, we can consider crystals with only translation symmetry ($C=\left\{I_n\right\}$). In this case, $\left(\mathbb T^n\rtimes\mathrm{SO}\left(n\right)\right)/C$ is just $\mathbb T^n\rtimes\mathrm{SO}\left(n\right)$.
    \end{remark}
    \begin{figure}
        \centering
        \includegraphics[scale=0.6]{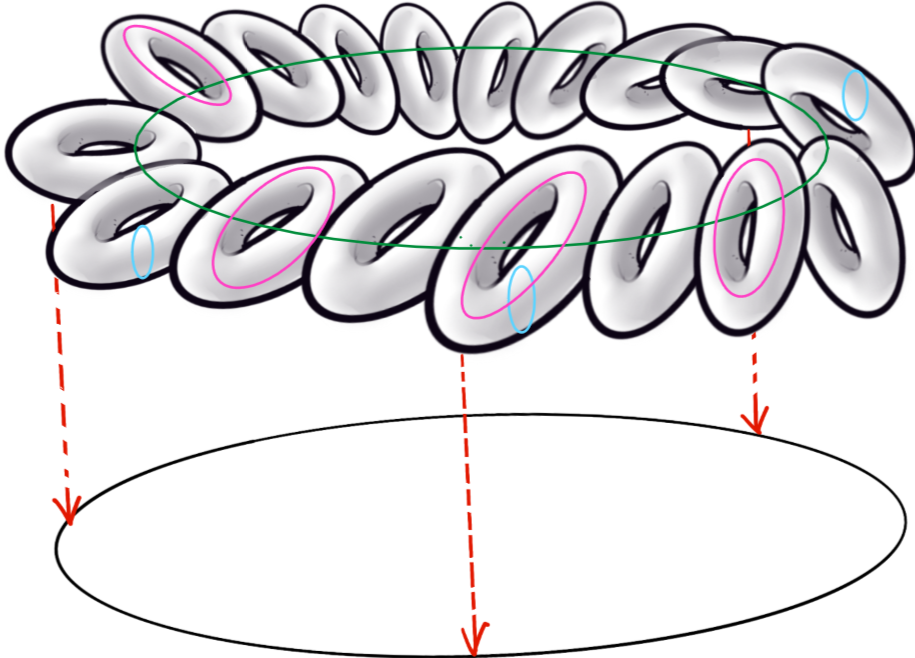}
        \caption{Space $\left(\mathbb T^2\rtimes\mathrm{SO}\left(2\right)\right)/C$ as a $\mathbb T^2$-bundle over $S^1$.}
        \label{fig:torusbundleovercircle}
    \end{figure}
    \subsection{\centering Homotopy groups}
    Let the universal cover of $\mathrm{SO}\left(n\right)$ be $q_n\colon\widetilde{\mathrm{SO}\left(n\right)}\to\mathrm{SO}\left(n\right)$. We have
    \begin{equation}
        \widetilde{\mathrm{SO}\left(n\right)}\cong
        \begin{cases}
            \mathbb R\quad\quad\quad\text{ if }n=2,\\
            \mathrm{Spin}\left(n\right)\text{ if }n>2.
        \end{cases}
    \end{equation}
    By Theorem~\ref{theo:covsemi}, the universal covering map is $\left(\mathrm{id}_{\mathbb R^n},q\right)\colon\mathbb R^n\rtimes\widetilde{\mathrm{SO}}\left(n\right)\to\mathbb R^n\rtimes\mathrm{SO}\left(n\right)$. Using Equation~\ref{eq:discretepihigh0} and the contractibility of $\mathbb R^n$ we can write
    \begin{equation}
        \pi_k\left(\mathbb R^n\rtimes\mathrm{SO}\left(n\right)/\mathbb A\rtimes C\right)\cong\pi_k\left(\mathbb R^n\rtimes\widetilde{\mathrm{SO}\left(n\right)}\right)\cong\pi_k\left(\widetilde{\mathrm{SO}\left(n\right)}\right),
    \end{equation}
    for $k>1$.
    We now compute the fundamental group. By Equation~\ref{eq:pi1univcovsub} we have
    \begin{equation}
         \pi_1\left(\mathbb R^n\rtimes\mathrm{SO}\left(n\right)/\mathbb A\rtimes C\right)\cong\left(\mathrm{id}_{\mathbb R^n},q_n\right)^{-1}\left(\mathbb A\rtimes C\right)=\mathbb A\rtimes q_n^{-1}\left(C\right)\cong\mathbb Z^n\rtimes q_n^{-1}\left(C\right).
    \end{equation}
    The action of $q_n^{-1}\left(C\right)$ on $\mathbb Z^n$ first factors through $q_n$ and is then given by the matrices $M_i$ (the coordinate representations of the generators of $C$). The action of $\pi_1$ on $\pi_n$ is moreover simplified as the free component of $\pi_1$ would have acted on the $\mathbb R^n$ component of $\pi_n$, which is contractible and hence we have the action of $\pi_1$ as just the action of $q_n^{-1}\left(C\right)$ on $\pi
    _n\left(\widetilde{\mathrm{SO}\left(n\right)}\right)$.
    
    \subsection{\centering Example: Point defects in the plane}
    Since $V_a$ is a singleton and $\mathbb R^2\setminus\left\{m\text{ points}\right\}\simeq\left(S^1\right)^{\vee m}$, we use Example~\ref{ex:wedgeofS1pi} and the preceding calculations to write
    \begin{equation}
        \begin{split}
            \mathrm{Def}_{\mathbb R^2}\left(m\text{ points}\right)\cong&\left(\text{conjugacy classes of }\mathbb Z^2\rtimes q_2^{-1}\left(C\right)\right)^m\\
            &\times
            \begin{cases}
                1\text{ if any reflection is a lattice symmetry},\\
                \mathbb Z/2\mathbb Z\text{ otherwise}.
            \end{cases}
        \end{split}
    \end{equation}
    In this case $q_2\colon\mathbb R\to\mathrm{SO}\left(2\right)\simeq S^1$, and $C\leq S^1$ is discrete (and hence must be cyclic of order $N=1,2,3,4,6$ for various lattices). Let the matrix representation of its generator be $M$. Then we have $q^{-1}_2\left(C\right)\cong\mathbb Z$ and $\mathbb Z^2\rtimes q_2^{-1}\left(C\right)\cong\mathbb Z^2\rtimes_M\mathbb Z$. To calculate $M$, we take a lattice, rotate its two basis vectors by $\frac{2\pi}N$, and express it again in the original basis. The coefficients for this will be $M$. Obviously, $M^N=I_2$ since $M$ is the coordinate representation of a rotation by $\frac{2\pi}N$.
    
    Each $\left(\utilde n,n_3\right)$ is an element $\pi_1$ which contains homotopy classes of maps. The map $p\circ\rho_{\utilde n,n_3}\circ\alpha$  is a representative element of the class $\left(\utilde n,n_3\right)$, where
    \begin{equation}
        \begin{split}
            \rho_{\utilde n, n_3}\colon S^1\to&\mathbb R^2\rtimes\mathrm U\left(1\right),\\
            t\mapsto&\left(t\left(n_1a_1+n_2a_2\right),e^{\frac{2\pi in_3t}N}\right),\\
            \alpha\colon\mathbb R^2\setminus\left\{\text{pt}\right\}\to&S^1,\\
            x\mapsto&\frac x{\left\Vert x\right\Vert},\\
            p\colon\mathbb R^2\rtimes\mathrm U\left(1\right)\to&G_0/\left(G_0\cap G_v\right),
        \end{split}
    \end{equation}
    and $a_1,a_2$ are the basis vectors of $\mathbb A$. The fact that these maps represent these classes is easily seen from Section~\ref{subsec:geometryRn} and Figure~\ref{fig:torusbundleovercircle}.
    
    Conjugation in the group $\mathbb Z^2\rtimes_M\mathbb Z$ is given by
    \begin{equation}
        \left(\utilde m,m_3\right)\left(\utilde n,n_3\right)\left(\utilde m,m_3\right)^{-1}=\left(\left(I_2-M^{n_3}\right)\utilde m+M^{m_3}\utilde n,n_3\right).
    \end{equation}
    Thus to form conjugacy classes we need to form the equivalence relation
    \begin{equation}
        \left(\utilde n,n_3\right)\sim\left(\left(I_2-M^{n_3}\right)\utilde m+M^{m_3}\utilde n,n_3\right),\text{ for any }\utilde m\in\mathbb Z^2\text{ and }m_3\in\mathbb Z.
    \end{equation}
    This relation was also obtained in \cite{Trebin82} but without any further applications. 
    Note that since conjugation does not change $n_3$, we can find the conjugacy classes for fixed $n_3$. That is, the set of conjugacy classes can be represented by
    \begin{equation}
        \coprod_{n_3\in\mathbb Z}F_{n_3},
    \end{equation} 
    where $F_{n_3}$ is class of defects at fixed $n_3$ (the disinclination index).
    
    \begin{table}
    \centering
    \footnotesize
    \begin{tabular}{|c|c|c|} 
        \hline
        Lattice&$M$&$F_{n_3}$\\ 
        \hline
        Parallelogram&$\begin{pmatrix}1&0\\0&1\end{pmatrix}$&$\mathbb Z^2$\\
        \hline
        Rectangle&$\begin{pmatrix}-1&0\\0&-1\end{pmatrix}$&$\left\{\utilde0\right\}\cup
        \begin{cases}
            \left\{\left(n_1,n_2\right)\;\vert\;n_2>0\right\}\cup\left\{\left(n_1,0\right)\;\vert\;n_1>0\right\}\text{ if }n_3\equiv0\mod2\\
            \left\{\left(0,1\right),\left(1,0\right),\left(1,1\right)\right\}\quad\quad\quad\quad\quad\quad\quad\;\text{ if }n_3\equiv1\mod2\\
        \end{cases}$\\
        \hline
        Square&$\begin{pmatrix}0&1\\-1&0\end{pmatrix}$&$\left\{\utilde0\right\}\cup
        \begin{cases}
            \left\{\left(n_1,n_2\right)\;\vert\;n_1\geq0,n_2>0\right\}\text{ if }n_3\equiv0\mod4\\
            \left\{\left(0,1\right)\right\}\quad\quad\quad\quad\quad\quad\quad\;\;\,\text{ if }n_3\equiv1,3\mod4\\  
            \left\{\left(0,1\right),\left(1,1\right)\right\}\quad\quad\quad\quad\quad\text{ if }n_3\equiv2\mod4
        \end{cases}$\\
        \hline
        Hexagonal&$\begin{pmatrix}1&1\\-1&0\end{pmatrix}$&$\left\{\utilde0\right\}\cup
        \begin{cases}
            \left\{\left(n_1,n_2\right)\;\vert\;n_1\geq0,n_2>0\right\}\text{ if }n_3\equiv0\mod6\\
            \emptyset\quad\quad\quad\quad\quad\quad\quad\;\,\quad\quad\quad\text{ if }n_3\equiv1,5\mod6\\  
            \left\{\left(0,1\right)\right\}\quad\quad\quad\quad\quad\quad\quad\;\;\,\text{ if }n_3\equiv2,3,4\mod6
        \end{cases}$\\
        \hline
    \end{tabular}
    \caption{Point defects in $\mathbb R^2$ at fixed disclination index.}
    \label{table:pointplane}
\end{table}
    
    We list matrices $M$ and classes $F_{n_3}$ for different lattice structures in Table~\ref{table:pointplane}. An example computation of $F_{n_3}$ for hexagonal lattices is given below.
    \subsubsection{Calculation of point defects in hexagonal lattice at fixed disclination index}
    We now classify the possible non-equivalent $n_1,n_2$ for fixed $n_3$. We note that $\left(M^{m_3}\utilde n,n_3\right)\sim\left(\utilde n,n_3\right)$ for any $m_3\in\mathbb Z$, and thus we need to only consider $\utilde n\in F$, where $F$ is the quotient $\mathbb Z^2/\sim$ (where $\utilde n\sim M^{m_3}\utilde n$) such that
    \begin{equation}
        F=\left\{\left(n_1,n_2\right)\;\vert\;n_1\geq0,n_2>0\right\}\cup\left\{\left(0,0\right)\right\}.
    \end{equation}
    Once we restrict ourselves to $\utilde n\in F$, the equivalence relation is simply (no two elements of $F$ are conjugate through multiplication by $M$)
    \begin{equation}
        \left(\utilde n,n_3\right)\sim\left(\left(I_2-M^{n_3}\right)\utilde m+\utilde n,n_3\right)\text{ for any }\utilde m\in\mathbb Z^2\text{ and }m_3\in\mathbb Z.
    \end{equation}
    We describe the results for various choices of $n_3$ in the following:
    \begin{enumerate}
        \item{$n_3\equiv0\mod6$}: In this case $M^{n_3}=I_2$, hence
        \begin{equation*}
            \left(\utilde n,n_3\right)\sim\left(\utilde n,n_3\right).
        \end{equation*}
        All the combinations $\left(\utilde n,0\right)$ are distinct for all $\utilde n\in F$.
        \item{$n_3\equiv1\mod6$}: In this case $M^{n_3}=M$. Since $I_2-M=M^5$, we write
        \begin{equation*}
            \left(\utilde n,n_3\right)\sim\left(M^5\utilde m+\utilde n,n_3\right).
        \end{equation*}
        Choosing $\utilde m=-M\utilde n$, we obtain that $\left(\utilde n,n_3\right)\sim\left(\utilde0,n_3\right)$ for any $\utilde n$. Thus there is only one conjugacy class here. We denote it by $\left(\utilde0,n_3\right)$.
        \item{$n_3\equiv2\mod6$}: In this case $I_2-M^{n_3}=\begin{pmatrix}1&-1\\1&2\end{pmatrix}$, which leads to
        \begin{equation*}
            \utilde n\sim\utilde n+m_1\left(1,1\right)+m_2\left(-1,2\right)\text{ for }m_1,m_2\in\mathbb Z.
        \end{equation*}
        implying that $\utilde n$ can take the values $\utilde0,~\left(0,1\right),~\text{or}~\left(0,2\right)$. Noting that 
        \begin{equation*}
            \begin{pmatrix}1&1\\-1&0\end{pmatrix}\begin{pmatrix}0\\2\end{pmatrix}-\begin{pmatrix}1\\1\end{pmatrix}= \begin{pmatrix}1&1\\-1&0\end{pmatrix}^2\begin{pmatrix}0\\1\end{pmatrix},
        \end{equation*}
        we write $\left(0,2\right)\sim\left(0,1\right)$. Thus the conjugacy classes are $\left(\utilde0,n_3\right)$ and $\left(\left(0,1\right),n_3\right)$.
        \item{$n_3\equiv3\mod6$}: In this case $M^{n_3}=-I_2$, hence
        \begin{equation*}
            \left(\utilde n,n_3\right)\sim\left(2\utilde m+M^{m_3}\utilde n,n_3\right).
        \end{equation*}
        As a result we are required to consider only those elements of $F$ for which both the components cannot be reduced further mod $2$, i.e., $\utilde0\text{ },\left(1,1\right),\left(0,1\right)$. Since $M\left(1,1\right)=\left(2,-1\right)\equiv\left(0,1\right)$ mod $2$, we have only two equivalence classes given by $\left(\utilde0,n_3\right)$ and $\left(\left(0,1\right),n_3\right)$.
        \item{$n_3\equiv4\mod6$}: An analysis similar to $n_3\equiv2\mod6$ implies that the only independent conjugacy classes are $\left(\utilde0,n_3\right)$ and  $\left(\left(0,1\right),n_3\right)$.
        \item{$n_3\equiv5\mod6$}: In this case $M^{n_3}=M^5=I_2-M$. The equivalence relation is
        \begin{equation*}
            \left(\utilde n,n_3\right)\sim\left(M\utilde m+\utilde n,n_3\right).
        \end{equation*}
        Following the arguments for the $n_3\equiv1\mod6$ case, we obtain that there is only one conjugacy class given by $\left(\utilde0,n_3\right)$.
    \end{enumerate}
    To summarise, the conjugacy classes at fixed $n_3$ are
    \begin{equation}
        F_{n_3}=\left\{\utilde0\right\}\cup
        \begin{cases}
            \left\{\left(n_1,n_2\right)\;\vert\;n_1\geq0,n_2>0\right\}\text{ if }n_3\equiv0\mod6,\\
            \emptyset\quad\quad\quad\quad\quad\quad\quad\;\,\quad\quad\quad\text{ if }n_3\equiv1,5\mod6,\\  
            \left\{\left(0,1\right)\right\}\quad\quad\quad\quad\quad\quad\quad\;\;\,\text{ if }n_3\equiv2,3,4\mod6.
        \end{cases}
    \end{equation}

\subsection{Example: Domain walls}
 Let us consider the defective set $X$ to be a union of $m$ $\left(n-1\right)$-dimensional hyperplanes in $\mathbb R^n$. Then $\mathbb R^n\setminus X$ will be a union of disjoint, contractible spaces (as each piece will be convex). For instance two parallel lines in $\mathbb R^2$ divide it into $3$ pieces, while intersecting lines will divide it into $4$. Therefore
\begin{equation}
    \mathbb R^n\setminus X\simeq M,
\end{equation}
where $M$ represents is a set. The defects will be given by
\begin{equation}
    \label{eq:domwall}
    \mathrm{Def}_{\mathbb R^n}\left(X\right)\cong\left(\coprod_{v\in V_a}\pi_0\left(G\right)/p\left(G_v\right)\right)^M,
\end{equation}
where we have used the fact that each $A\in\pi_o\left(\mathbb R^n\setminus X\right)$ is contractible. Hence
\noindent$\mathrm{hTop}\left[A,G_0/\left(G_0\cap G_v\right)\right]\cong\pi_0\left(G_0/\left(G_0\cap G_v\right)\right)$ which is $1$ using path-connectedness of 
\noindent$G_0/\left(G_0\cap G_v\right)$. In the context of material science, Equation~\ref{eq:domwall} tells us that such defects are equivalent to choosing cardinality of $M$ worth of crystal structures and a chirality for each crystal structure, one structure for each disconnected contractible piece. These are commonly known as domain walls.

\subsection{\centering Example: Textures}
A texture is a global defect where $X$, the defective subset of the manifold, is empty but the map $M \to V$ is not homotopy to the trivial constant map~\cite{Finkelstein66}. 
In flat crystals, we get nontrivial textures when boundary conditions are imposed. We see this in the following two computations. 

First we calculate the set of textures on $\mathbb R^n$ as  
\begin{equation}
    \mathrm{hTop}\left[\mathbb R^n,G_0/\left(G_0\cap G_v\right)\right]\times\frac{\mathbb Z/2\mathbb Z}{p\left(G_v\right)}\cong\pi_0\left(G_0/\left(G_0\cap G_v\right)\right)\times\frac{\mathbb Z/2\mathbb Z}{p\left(G_v\right)}\cong\frac{\mathbb Z/2\mathbb Z}{p\left(G_v\right)},
\end{equation}
where in the first equality we use the contractibility of $\mathbb R^n$ and in the second we use the fact that $G_0/\left(G_0\cap G_v\right)$ is connected.
Second, we take boundary conditions such that the deformation of a crystal due to a defect vanishes at infinity. As explained in Remark~\ref{rem:bdry}, we can incorporate such a condition by considering maps out of the $1$-point compactification of $\mathbb R^n$, i.e., $S^n$. The textures satisfying the boundary conditions are given by
\begin{equation}
    \mathrm{hTop}\left[S^n,G_0/\left(G_0\cap G_v\right)\right]\times\frac{\mathbb Z/2\mathbb Z}{p\left(G_v\right)}\cong\frac{\pi_n\left(G_0/\left(G_0\cap G_v\right)\right)}{\pi_1\left(G_0/\left(G_0\cap G_v\right)\right)}\times\frac{\mathbb Z/2\mathbb Z}{p\left(G_v\right)}\cong\frac{\pi_n\left(\widetilde{\mathrm{SO}\left(n\right)}\right)}{q_n^{-1}\left(C\right)}\times\frac{\mathbb Z/2\mathbb Z}{p\left(G_v\right)}.
\end{equation}
In the plane, only two topologically distinct textures are possible if the lattice has chirality (no reflection symmetry). There are no nontrivial textures otherwise as $\pi_2\left(\mathbb R^2\right)\cong0$. However, they are possible in $\mathbb R^3$ even for achiral lattices as $\pi_3\left(\mathrm{Spin}\left(3\right)\right)\cong\mathbb Z$.  If we ignore boundary conditions, only chirality based textures are possible in any dimension.
In Table~\ref{table:cohomology}, as we shall see below, nontrivial textures can also arise for non-contractible manifolds like cylinder, tori, and annulus, even without explicit boundary conditions.

\section{\centering Crystals in spheres}
\label{sec:crystalsS}
For crystals in $S^n$ we consider, in accordance with Section~\ref{sec:crystals}, 
\begin{equation}
    G=\mathrm{Isom}\left(S^n\right)\cong\mathrm O\left(n+1\right),
\end{equation}
where $\mathrm O\left(n+1\right)$ acts on $S^n$ in the natural way. Hence $\pi_0\left(G\right)\cong\mathbb Z/2\mathbb Z$ and $G_0\cong\mathrm{SO}\left(n+1\right)$. Accordingly $\pi_0\left(G\right)/p\left(G_v\right)$ is either trivial or $\mathbb Z/2\mathbb Z$ depending on whether $G_v$ contains a reflection or not. Also, as discussed in Section~\ref{sec:crystals}, we have to look for discrete cocompact subgroups $\Gamma=_\text{def}G_0\cap G_v$ of $\mathrm{SO}\left(n+1\right)$. Since $S^n$ is already compact, we just need to look at discrete $\Gamma$ for it will be automatically cocompact.
\subsection{\centering Homotopy groups}
    By Equation~\ref{eq:discretepihigh0} we have, for $k>1$, 
    \begin{equation}
        \pi_k\left(\mathrm{SO}\left(n+1\right)/\Gamma\right)\cong\pi_k\left(\mathrm{Spin}\left(n+1\right)\right),
    \end{equation}
    where the fact that $\widetilde{\mathrm{SO}\left(n+1\right)}=\mathrm{Spin}\left(n+1\right)$, for $n>1$, has been used. On the other hand, we use Equation~\ref{eq:pi1univcovsub} to write
    \begin{equation}
         \pi_1\left(\mathrm{SO}\left(n+1\right)/\Gamma\right)\cong q_{n+1}^{-1}\left(\Gamma\right),
    \end{equation}
    where $q_{n+1}\colon\mathrm{Spin}\left(n+1\right)\to\mathrm{SO}\left(n+1\right)$ is as defined in the previous section.
    \subsection{\centering Example: point defects in the 2-sphere}
    We have
    \begin{equation*}
        S^2\setminus\left\{m\text{ points}\right\}\simeq
        \begin{cases}
            S^2\text{ if }m=0,\\
            \left(S^1\right)^{\vee m-1}\text{ if }m>0.
        \end{cases}
    \end{equation*}
    Therefore, for $m=0$, we obtain
    \begin{equation}
        \begin{split}
            \mathrm{Def}_{S^2}\left(\emptyset\right)&\cong\mathrm{hTop}\left[S^2,\mathrm{SO}\left(3\right)/\Gamma\right]\times\frac{\mathbb Z/2\mathbb Z}{p\left(\Gamma\right)}\cong\frac{\pi_2\left(\mathrm{SO}\left(3\right)/\Gamma,\Gamma\right)}{\pi_1\left(\mathrm{SO}\left(3\right)/\Gamma,\Gamma\right)}\times\frac{\mathbb Z/2\mathbb Z}{p\left(\Gamma\right)}\\
            &\cong\frac{\pi_2\left(\mathrm{SO}\left(3\right),\Gamma\right)}{q_3^{-1}\left(\Gamma\right)}\times\frac{\mathbb Z/2\mathbb Z}{p\left(\Gamma\right)}\cong\frac{\mathbb Z/2\mathbb Z}{p\left(\Gamma\right)},
        \end{split}
    \end{equation}
    where we have used Equation~\ref{eq:htpygroupsarise} in the second step and the homotopy groups calculation from the previous section in the third. In the fourth step we used the fact that $\pi_2$ of any Lie group is trivial. For $m\geq1$, we recall Example~\ref{ex:wedgeofS1pi} to write
    \begin{equation}
        \mathrm{Def}_{S^2}\left(\left\{m\text{ points}\right\}\right)\cong\left(\text{conjugacy classes of }q_3^{-1}\left(\Gamma\right)\right)^{m-1}\times\frac{\mathbb Z/2\mathbb Z}{p\left(\Gamma\right)}.
    \end{equation}
    In the above relationships, $q_3^{-1}\left(\Gamma\right)$ is a discrete subgroup of $\mathrm{Spin}\left(3\right)$. Accordingly we only classify such discrete subgroups and find their conjugacy classes. This is achieved by using the exceptional isomorphism $\mathrm{Spin}\left(3\right)\cong\mathrm{SU}\left(2\right)$, and the ADE classification of discrete subgroups of $\mathrm{SU}\left(2\right)$. This classification of subgroups, and the number of associated conjugacy classes, is listed in Table~\ref{table:pointsphere}, where the number of conjugacy classes have been computed using the quaternion representation of $SU\left(2\right)$. An analogous  classification of two dimensional spherical crystals by lattice types was discussed in~\cite{Kleman89} although without an explicit calculation of the number of conjugacy classes.
    
\begin{table}
    \small
    \centering
    \begin{tabular}{| c | c | c | c | c |} 
        \hline
        Group & Point Group & Order & Angles & Number of Conjugacy Classes\\\hline
        Binary Cyclic: $A_n$&$\left(1,n,n\right)$&$2n$&$\frac{2\pi}n$&$n$\\\hline
        Binary Dihedral: $D_n$&$\left(2,2,n\right)$&$4n$&$\frac{2\pi}n,\pi$&$n+3$\\ \hline
        Binary Tetrahedral: $E_6$&$\left(2,3,3\right)$&$24$&$\frac{2\pi}3,\pi$&$7$\\\hline
        Binary Octahedral: $E_7$&$\left(2,3,4\right)$&$48$&$\frac{2\pi}3,\frac{\pi}2,\pi$&$9$\\\hline
        Binary Icosahedral: $E_8$&$\left(2,3,5\right)$&$120$&$\frac{2\pi}5,\frac{2\pi}3,\pi$&$11$\\\hline
    \end{tabular}
    \caption{The binary polyhedral groups (discrete subgroups of $\mathrm{SU}\left(2\right)$) and their conjugacy classes.}
    \label{table:pointsphere}
\end{table}

\section{\centering Cohomological Examples}
\label{sec:crystalsC}

Having previously discussed crystals on homogeneous spaces, both Euclidean and spherical, we now study defects in crystals on more general manifolds. 
Suppose $G_0\cong\mathbb R^n\times\mathbb T^m$. Since $G_v$ is discrete for crystals, it is closed in $G$. Therefore $G_v\cap G_0$ is closed in $G_0$ and is hence a cocompact closed subset of an abelian Lie group. The quotient $G_0/\left(G_v\cap G_0\right)$ will then be a compact abelian Lie group of the same dimension as $G_0$, i.e., the torus $\mathbb T^{n+m}$.
We can thereafter use Example~\ref{ex:cohomology} to evaluate the set of defects as
\begin{equation}
    \label{eq:coho}
    \mathrm{Def}_M\left(X\right)\cong\prod_{A\in\pi_0\left(M\setminus X\right)}H^1_\mathrm{sing}\left(A,\mathbb Z\right)^{n+m}\times\pi_0\left(G\right)/p\left(G_v\right).
\end{equation}

Some systems for which $G_0/\left(G_v\cap G_0\right)$ is a torus are given in Table~\ref{table:cohomology}. In each of these cases $X$ is a set of $m$ points. The relevant figures for the deformation retracts are also given. The set of defects is calculated using Equation~\ref{eq:coho}. We can easily compute $\pi_0\left(G\right)/p\left(G_v\right)$ since $p\left(G_v\right)$ is a subgroup of $\pi_0\left(G\right)$ which are classified in a straightforward manner. For a 2D cylinder, $2$-Torus, and an infinite 2D flat annulus embedded in $\mathbb{R}^3$, the metric is induced from the embedding and the corresponding isometry groups are mentioned as $G$. We have also considered an $n$-dimensional flat torus, obtained by quotienting $\mathbb{R}^n$ by a lattice, whose isometry group consists of translations (quotiented by the lattice) and those rotations that preserve the lattice. A 2D flat torus is diffeomorphic to the regular torus but not isometric to it. An explicit fractal-like $C^1$ isometric embedding of the flat torus in $\mathbb{R}^3$ has been found~\cite{Borrelli7218}. An important distinction in the defects structure of the non-contractible manifolds considered here is that even when there are no point singularities, the case where the defective set is empty, i.e., $X=\emptyset$, $\mathrm{Def}_M\left(X\right)$ is still non-trivial. There can be global textures on these manifolds. 

\begin{table}
    \centering
    \footnotesize
    \begin{tabular}{||c||c|c|c|c|} 
    \hline
        Examples&\textbf{$2$D Cylinder}&\textbf{$2$D Torus}&\textbf{$n$D Flat Torus}&\textbf{$2$D Annulus}\\
        \hline
        $M$&\includegraphics[width=0.1675\textwidth]{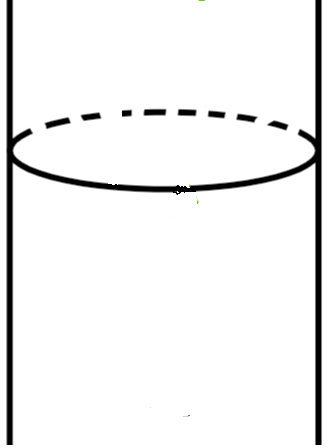}&\includegraphics[width=0.11\textwidth]{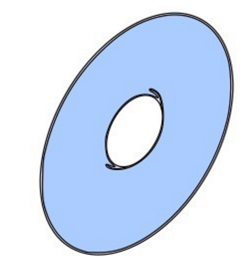} & \includegraphics[width=0.1675\textwidth]{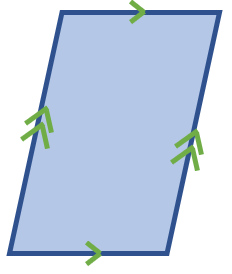} &\includegraphics[width=0.1675\textwidth]{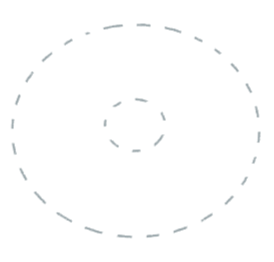} \\
        &$\mathbb R\times S^1$ in $\mathbb R^3$&$\mathbb T^2$ in $\mathbb R^3$&$\mathbb R^n/\Lambda$ (Flat Tori)&$S^1\times\left(0,\infty\right)$ in $\mathbb R^2$\\\hline
        $G$&$\left(\mathbb R\times S^1\right)\rtimes V_4$&$S^1\rtimes V_4$&$\mathbb T^n\rtimes\mathrm{Aut}\left(\Lambda\right)$&$S^1\rtimes\mathbb Z/2\mathbb Z$\\\hline
        $G_0$&$\mathbb R\times S^1$&$S^1$&$\mathbb T^n$&$S^1$\\\hline
        $\pi_0\left(G\right)$&$V_4$&$V_4$&$\mathrm{Aut}\left(\Lambda\right)$&$\mathbb Z/2\mathbb Z$\\\hline
        $M\setminus X$&\includegraphics[width=0.1675\textwidth]{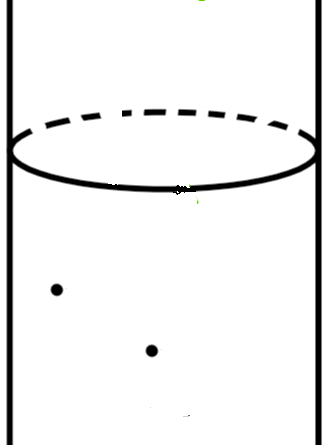}&\includegraphics[width=0.11\textwidth]{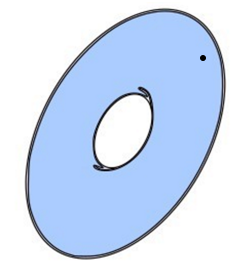} &\includegraphics[width=0.1675\textwidth]{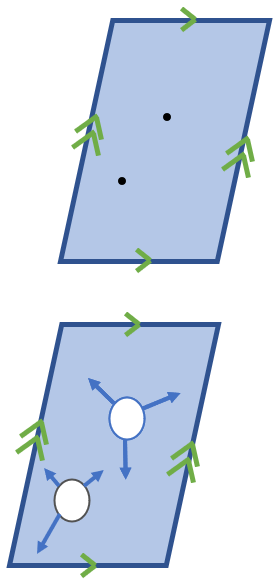}&\includegraphics[width=0.1675\textwidth]{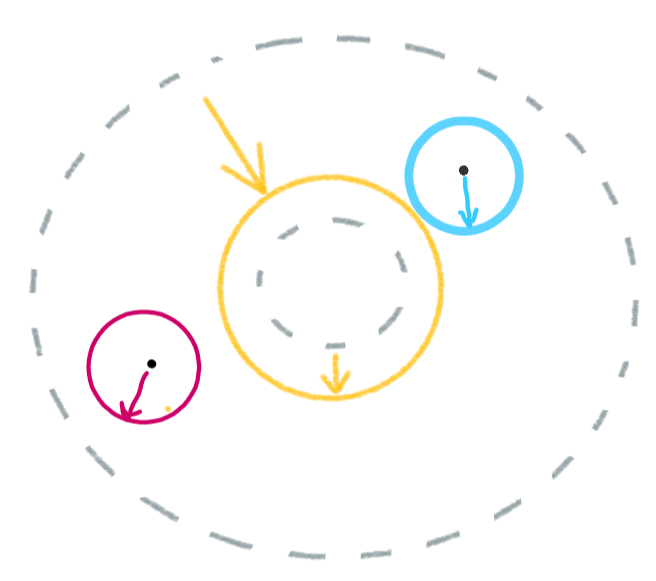} \\
        &\includegraphics[width=0.1675\textwidth]{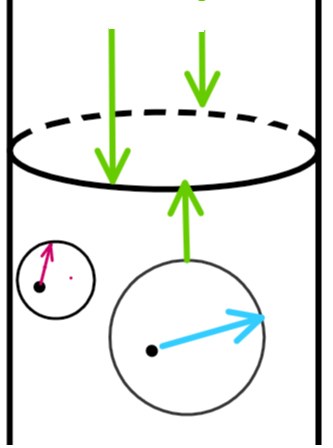}&\includegraphics[width=0.11\textwidth]{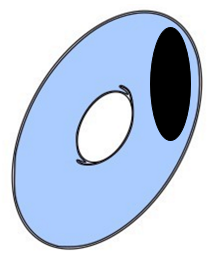}&\includegraphics[width=0.1675\textwidth]{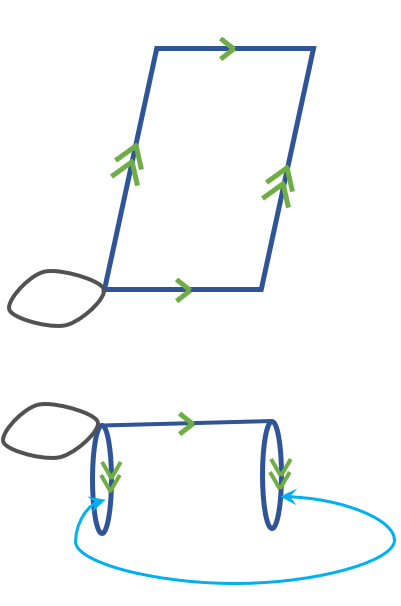}& \\
        &\includegraphics[width=0.1675\textwidth]{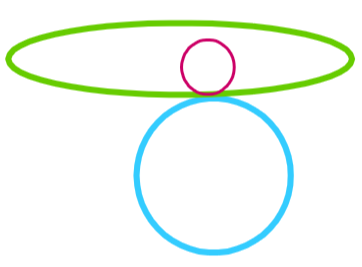}&\includegraphics[width=0.11\textwidth]{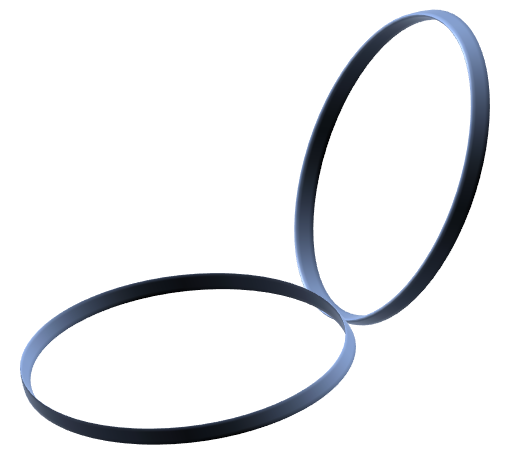} &\includegraphics[width=0.1675\textwidth]{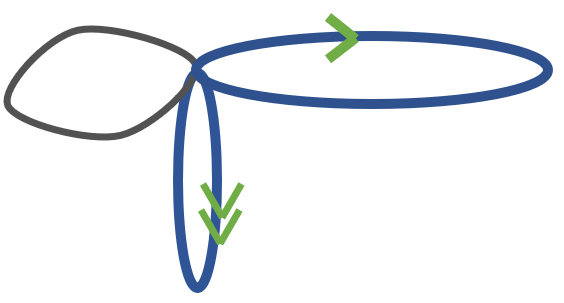}&\includegraphics[width=0.1675\textwidth]{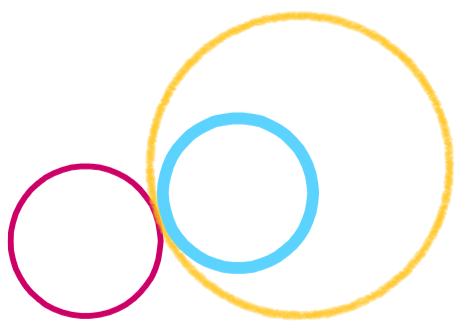}\\
        &$\left(S^1\right)^{\vee m+1}$&$\begin{cases}
            m=0\colon\mathbb T^2\\
            \left(S^1\right)^{\vee m+1}
        \end{cases}$&$\begin{cases}
            m=0\colon\mathbb T^n\\
            \left(S^{n-1}\right)^{\vee m+1}
        \end{cases}$&$\left(S^1\right)^{\vee m+1}$\\\hline
        $H^1_\text{sing}\left(M\setminus X,\mathbb Z\right)$&$\mathbb Z^{m+1}$&$\mathbb Z^{m+1+\delta_{m0}}$&$\begin{cases}
            m=0\colon\mathbb Z^n\\
            m\neq0,n\neq2\colon0\\
            \mathbb Z^{m+1+\delta_{m0}}
        \end{cases}$
        &$\mathbb Z^{m+1}$\\\hline
        
    \end{tabular}
    \caption{Examples of crystals for which $G_0/\left(G_0\cap G_v\right)$ is a torus.}
    \label{table:cohomology}
\end{table}

\section{\centering Physical Interpretation}
\label{sec:phyjust}
Given the setup of Section~\ref{sec:introchaptopdef}, we are interested in calculating the partition function (which we can further use to calculate the statistical properties of the system)~\cite[Chapter 6]{Sethna21},
\begin{equation}
    \label{eq:partfun}
    Z_M\left(\beta\right)=\int_{\mathcal C_M}e^{-\beta\mathcal H_M}.
\end{equation}
The partition function can be evaluated by perturbing around the minima of $\mathcal H_M$ and making a saddle-point approximation. To compute the path integral we need to sum over the perturbations around the minimum from each path-connected component. For the saddle-point to make sense, we would want to look at path-connected components of $\mathcal C_M$. This is inspired from the technique of summing over instantons~\cite{nlab:instanton} to obtain non-perturbative corrections to the partition function in QFT.  
In other words, we want to split Equation~\ref{eq:partfun} as
\begin{equation}
    Z_M\left(\beta\right)=\sum_{X\in\pi_0\left(C_M\right)}\int_Xe^{-\beta\mathcal H_M}.
\end{equation}
Therefore, to evaluate the partition function, we must understand $\pi_0\left(\mathcal C_M\right)$ as well as the space of minima of $\mathcal H_M$, i.e., $V_M$. Next, we turn to $\pi_0\left(\mathcal C_M\right)$, which is a more complex object. In order to move further, we will require some simplifying assumptions.
\begin{definition}
    Given a partial function $f\colon M\rightharpoonup N$, let $D_f\subseteq M$ to be the subset on which it is not defined. If $M,N$ are topological spaces, we call $f$ continuous if $f\vert_{M\setminus D_f}\colon M\setminus D_f\to N$ is continuous. Let $\mathrm{Top_{part}}$ be the category of topological spaces and continuous partial functions.
\end{definition}
\begin{definition}
    \label{def:homotopypartfunc}
    For two continuous partial functions $f,g\colon M\rightharpoonup N$, a homotopy from $f$ to $g$ is a pair $\left(\alpha,H\right)$ where  $\alpha\colon M\setminus D_f\xrightarrow{\sim} M\setminus D_g$ is a homotopy equivalence of spaces and $H$ is a homotopy from $f\vert_{M\setminus D_f}$ to $g\vert_{M\setminus D_g}\circ\alpha$. That is, the following diagram commutes up to the homotopy $H$,
    \begin{equation}
        \begin{tikzcd}
            M\setminus\mathrm D_f\arrow[rd,"H",Rightarrow]\arrow[dd,"\alpha"',"\simeq"]\arrow[rr,"f\vert_{M\setminus D_f}"]&&N\\&\text{ }&\\
            M\setminus D_g\arrow[rruu,"g\vert_{M\setminus D_g}"']
        \end{tikzcd}
    \end{equation}
    Let $\mathrm{hTop_{part}}$ denote the category with topological spaces as objects and morphisms defined by
    \begin{equation}
        \mathrm{hTop_{part}}\left[M,N\right]=\mathrm{Top_{part}}\left[M,N\right]/\text{homotopy}.
    \end{equation}
    There is an obvious functor $Q\colon\mathrm{Top_{part}}\to\mathrm{hTop_{part}}$ which identifies morphisms if they are homotopic.
\end{definition}
\begin{assumption}
    \label{assump:main}
    $\pi_0\left(\mathcal C_M\right)\cong\mathrm{hTop_{part}}\left[M,V_M\right]$.
\end{assumption}
\noindent This equivalence is the key to the interpretation of topological classification of defects in the sense introduced in Definition~\ref{def:def}. In order to justify Assumption~\ref{assump:main}, we relate it to a set of physically motivated assumptions below. Following Remark~\ref{rem:notation}, we will suppress the $M$ dependence whenever it is clear from the context at hand.

\noindent \textbf{A1} \textit{For any point $p\in M$, we have a function $d_p\colon\mathcal C\times\mathcal C\to\mathbb R_{\geq0}$ which measures the `closeness' of two configurations `around' $p$.}

 \noindent   This can be used to introduce a `localising' function $\mathrm{Loc}\colon\mathcal C\to\mathrm{Set_{part}}\left[M,V\right]$ such that
    \begin{equation}
      \label{eq:loc}
        \mathrm{Loc}_\phi\left(p\right)=\operatorname{\mathrm{arg\,min}}_{v\in V}d_p\left(\phi,v\right)\text{ where }\phi\in\mathcal C,p\in M.
    \end{equation}
 Hence, given a configuration (for crystals, a deformed lattice), we assign to each point of the manifold the vacuum (for crystals, a perfect lattice) which is closest to the given configuration at that point. For some points $p$, there may not be a unique $v\in V$ which minimises $d_p\left(\phi,v\right)$. At such points we leave $\mathrm{Loc}_\phi\left(p\right)$ undefined.

\noindent  \textbf{A2} \textit{For any $\phi\in\mathcal C$, $\mathrm{Loc}_\phi\left(p\right)$ is a continuous partial function.}

  \noindent  Thus, we actually have a function $\mathrm{Loc}\colon\mathcal C\to\mathrm{Top_{part}}\left[M,V\right]$.
    
\noindent \textbf{A3} \textit{For any $g\in\mathrm{Top_{part}}\left[M,V\right]$, there is a $\phi\in\mathcal C$ such that there is a homotopy of partial functions from $\mathrm{Loc}_\phi$ to $g$.}

\noindent This is physically reasonable because, given the structure of local configurations, we should be able to create a global configuration that is deformable to the original one. This assumption implies that $Q\vert_{\mathrm{Top_{part}}\left[M,V\right]}\circ\mathrm{Loc}\colon\mathcal C\to\mathrm{hTop_{part}}\left[M,V\right]$ is surjective, where $Q$ is the homotopy quotient functor defined above.
    
\noindent \textbf{A4} \textit{For any $\phi_1,\phi_2\in\mathcal C$}
        \begin{equation}
            \exists \text{ a homotopy }\left(\alpha,H\right)\colon\mathrm{Loc}_{\phi_1}\xRightarrow{\left(\alpha,H\right)}\mathrm{Loc}_{\phi_2}\iff\text{there is a path from }\phi_1\text{ to }\phi_2\text{ in }\mathcal C.
        \end{equation}
   
 \noindent   Accordingly, two configurations which are deformable to each other by perturbations, i.e., they are path connected in the configuration space, must also be locally deformable to each other. Here, the notion of locally deformable is captured by the existence of a homotopy of partial functions between $\mathrm{Loc}_{\phi_1}$ and $\mathrm{Loc}_{\phi_2}$. Conversely, if two configurations are locally deformable to each other, then one can reach the other through perturbations. The $\Leftarrow$ and $\Rightarrow$ parts of Assumption \textbf{A4} imply, respectively, that $Q\vert_{\mathrm{Top_{part}}\left[M,V\right]}\circ\mathrm{Loc}$ descends to a map $\Phi\colon\pi_0\left(\mathcal C\right)\to\mathrm{hTop_{part}}\left[M,V\right]$ and that $\Phi$ in injective. Since we already had surjectivity, this gives a proof of Assumption~\ref{assump:main} from a set of physically motivated assumptions. Moreover, Assumptions \textbf{A3} and \textbf{A4} guarantee that the elements of $\mathcal C$ are accurately represented by continuous partial maps $M\rightharpoonup V$ (with maybe some continuous deformations).  

\begin{definition}
    For a space $M$, we define its homotopy powerset as
    \begin{equation}
        \mathrm h\mathcal P\left(M\right)=\mathcal P\left(M\right)/\sim,~\text{where }X_1,X_2\subseteq M\text{ satisfy }X_1\sim X_2\text{ if }M\setminus X_1\simeq M\setminus X_2.
    \end{equation}
\end{definition}
\begin{theorem}[Relationship between defects and components of the configuration space]
    \label{theo:pi0canddefrel}
    \begin{equation}
        \pi_0\left(\mathcal C_M\right)=\coprod_{\tilde X\in \mathrm h\mathcal P\left(M\right)}\mathrm{Def}_M\left(\tilde X\right)
    \end{equation}
\end{theorem}
\begin{proof}
    Note that we can write $\mathrm{Top_{part}}\left[M,V\right]$ as
    \begin{equation}
        \mathrm{Top_{part}}\left[M,V\right]\cong\coprod_{X\subseteq M}\mathrm{Top}\left[M\setminus X,V\right].
    \end{equation}
    If we have a homotopy equivalence $\alpha\colon M\setminus X_1\xrightarrow{\sim}M\setminus X_2$ then from Definition~\ref{def:homotopypartfunc} it follows that the sets $\mathrm{Top}\left[M\setminus X_1,V\right]$ and $\mathrm{Top}\left[M\setminus X_2,V\right]$ have the same image under $Q$. The image is $\mathrm{hTop}\left[M\setminus X_1,V\right]$, where $X_1$ can be replaced by $X_2$ as the result is independent of this choice. The identification
    \begin{equation}
        \mathrm{hTop_{part}}\left[M,V\right]=\coprod_{\tilde X\in \mathrm h\mathcal P\left(M\right)}\mathrm{Def}_M\left(X\right),
    \end{equation}
    where $X$ is any representative of $\tilde X$, combined with Assumption~\ref{assump:main}, yields the desired result.
\end{proof}
We observe, in particular, that choosing an element of $\pi_0\left(\mathcal C\right)$ involves a choice of a defective subset (up to homotopy) and a choice of defect from $\mathrm{Def}_M\left(X\right)$. Substituting Theorem~\ref{theo:defclass} in Theorem~\ref{theo:pi0canddefrel}, we obtain
\begin{equation}
    \pi_0\left(\mathcal C\right)\cong\coprod_{\tilde X\in \mathrm h\mathcal P\left(M\right)}\prod_{A\in\pi_0\left(M\setminus X\right)}\coprod_{v\in V_a}\mathrm{hTop}\left[A,G_0/\left(G_0\cap G_v\right)\right]\times\pi_0\left(G\right)/p\left(G_v\right).
\end{equation}
    Computing $\mathrm h\mathcal P\left(M\right)$ is extremely difficult for commonly encountered manifolds. We are therefore able to calculate $\mathrm{Def}_M\left(X\right)$ only for certain special cases of simple subsets $X\subseteq M$. This choice gives us the expression for classification of topological defects, defective on $X$, as given in Equation~\ref{eq:expandedexpression}.

\section{\centering Critique}
\label{sec:critique}
There are inconsistencies in the traditional interpretation of the classification results. For instance, when we identify the possible point defects in $\mathbb R^2$ (for any lattice shape), we obtain classes of triples of integers. In the literature these are interpreted in terms of components of a dislocation's Burgers vector and a disclination index. However, we note that not all triples represent distinct configurations. The failure for all triples to remain distinct arises from the semidirect product, or the twisted bundle nature, of the isometry group (and the space of global lattices). Moreover, arbitrary loops cannot be assigned triples of integers (we observe that only conjugacy classes have physical meaning). Therefore even if we tried to interpret the results in an ad hoc way, e.g., by identifying triples with the winding numbers, we may be lead to failure. Indeed, we know that there is a homotopy between loops within the same conjugacy class but how this translates to deformations in a lattice is ambiguous due to the gluing problem discussed next.

Given a map into the space of global lattices (order parameter space), it is unclear how to glue them together to form a deformed lattice. This occurs due to the global nature of the order parameter space, which derives from taking $G$ to be $\mathrm{Isom}(M)$ and the local nature of the gluing operation.
For general manifolds, there are both basepoint-preserving isometries (rotations in $\mathbb R^n$) and isometries that shift every point (translations in $\mathbb R^n$). Capturing both of them in the isometry group (and its quotient) hinders the local nature of gluing. However, both of these isometries are essential to describing general defects.
\begin{equation}
    \begin{tikzcd}
        \text{Deformed lattices}\arrow[r,bend left=20,"\text{Localising}"]&M\to\text{Global lattices}\arrow[l,bend left=20,"\text{Gluing}"]
    \end{tikzcd}
    \label{eq:locGlu}%
\end{equation}
 As demonstrated in Equation~\ref{eq:locGlu}, parallel to the gluing process, we have the process of localising, i.e., assigning a locally perfect lattice to every point of the manifold. Given any configuration $\phi \in \mathcal{C}_M$, which can be thought of as a deformed lattice, we have defined a map $\mathrm{Loc}_{\phi}$, in Equation~\ref{eq:loc}, which assigns a local lattice to each point of the manifold. Since lattices are discrete, to define `around' we need to choose a finite open set around $p$, a choice which is not canonical. If, however, we have a continuum limit, then this problem could be resolved using germs. Nevertheless by working only in the continuum limit we risk losing some defect data. Moreover, as in the case of certain compact manifolds (e.g., $S^2$), all the lattices are necessarily finite and there may not be any way to directly compute the continuum limit. 
           
Another inconsistency is noted by recalling that quotients by discrete subgroups are not always perfect crystals. Suppose we want to put a hexagonal lattice on $S^2$. Then, by classification theory, we should look for an $H\leq\mathrm{Isom}\left(S^2\right)\cong\mathrm O\left(3\right)$ (discrete with compact coset space) whose fundamental domain is a regular $6$ sided figure. However, we know that such an $H$ does not exist. An alternative could be to consider the tiling of $S^2$ by hexagons interspersed with $12$ pentagons (icosahedral symmetry). We would like to think of the pentagons as `defects' in a pure hexagonal lattice, but the classification process would count the above as a proper lattice as it has a discrete symmetry group and compact fundamental domain.

It is possible to have additional (non-isometry) symmetries of the Hamiltonian, such as dilation invariance, conformal invariance or additional internal degrees of freedom which are simply not captured in this formalism, because we always restrict ourselves to look at symmetries of the Hamiltonian which are subgroups of the isometry group. Another facet of this problem is that the classification procedure becomes highly sensitive to even the slightest perturbations of the manifold. Indeed, even a small perturbation of a symmetric space makes it almost entirely asymmetric. This rigidity is not consistent with our physical intuition. 
Finally, we note that there could have been additional energetic or boundary restrictions on the order parameter space which can prevent the physical realisation of all classes of defects or may result into a splitting of the defect classes.

Some of the critiques presented here are not in exact agreement with the assumptions, presented in Section~\ref{sec:phyjust}, which justify the rationale behind the classification methodology for defects in crystals. The Definition~\ref{def:def} therefore should be used for crystal lattices with care and keeping in mind the assumptions outlined in the preceding section as well as the critique presented above. Many of our objections, and the attempt to clearly state the underlying assumptions, are motivated in order to rigorize and extend critical arguments present in earlier works \cite{Mermin79}, \cite{Finkelstein66}, \cite{MG80}, \cite{Balian1981PhysiqueDD}. 


\section{\centering Reconciliation}
\label{reconcile}
We discuss two different extensions of the theory developed so far in order to address some of the concerns which were raised above. What follows is, however, mostly a collection of ideas which would need to be made precise in further communications. 
\subsection{What are disclinations?}
The classification of line defects in crystals using homotopy theory is ambiguous in distinguishing between what are traditionally called dislocations and disclinations. The breaking of rotational symmetry, identified with disclinations, has a unique meaning for simple manifolds, such as Euclidean spaces and spheres, but for a more general manifold isometry group may not always be separable into rotations and translations. Therefore the consideration of defects separately as dislocations and disinclination breaks down. 
Instead, we propose to define disclinations in terms of rotation groups of a point, i.e., the little group of the point.
\subsubsection{Rotation group of a point}
For a Riemannian manifold $M$, with metric $g$, the group of isometries is a Lie group which acts smoothly on $M$ \cite{MS39}. We consider the little group of a point $p\in M$, i.e., the stabiliser
\begin{equation}
    \mathrm{Fix}\left(p\right)=_\text{def}\mathrm{Stab}_{\mathrm{Isom}\left(M\right)}\left(p\right)\leq\mathrm{Isom}\left(M\right).
\end{equation}
Clearly, $\mathrm{Fix}\left(p\right)$ is closed and, by Cartan's closed subgroup theorem~\ref{theo:cartan}, a Lie group with manifold structure, in agreement with its embedding into the isometry group.

Since any $f\in\mathrm{Fix}\left(p\right)$ is an isometry and $f\left(p\right)=p$, it induces a linear map
\begin{equation}
    \mathrm df\vert_p\colon T_pM\to\mathrm T_pM,
\end{equation}
which preserves the metric. Differentiating at $p$ gives us a homomorphism
\begin{equation}
    \label{eq:homFixOrth}
    \mathrm{Fix}\left(p\right)\to\mathrm O\left(T_pM,g_p\right),
\end{equation}
where, for an inner product space $\left(V,g\right)$ over $\mathbb R$, $O\left(V,g\right)$ is the orthogonal group consisting of linear maps $V\to V$ which preserve the inner product $g$. 
We now define the `rotation group' at a point to be the path component containing the identity of $\mathrm{Fix}\left(p\right)$, i.e.,
\begin{equation}
    \mathrm{Rot}\left(p\right)=\left(\mathrm{Fix}\left(p\right)\right)_0.
\end{equation}
We define $\mathrm{SO}\left(V,g\right)$ as the component of $\mathrm O\left(V,g\right)$ containing the identity. Then the homomorphism from Equation~\ref{eq:homFixOrth} descends into a homomorphism
\begin{equation}
    \theta_p\colon\mathrm{Rot}\left(p\right)\to\mathrm{SO}\left(T_pM,g_p\right).
\end{equation}
This homomorphism is injective. For $f\in\mathrm{Rot}\left(p\right)$, we call $\theta_p\left(f\right)$ its `angle'.
\subsubsection{Single disclination}
Let $\Gamma\leq\mathrm{Isom}\left(M\right)$ be a discrete subgroup such that the coset space $\mathrm{Isom}\left(M\right)/\Gamma$ is compact (recall that the quotient of a Lie group by a closed subgroup is a manifold). The space of lattices with $\Gamma$ as their symmetry group is $\mathrm{Isom}\left(M\right)/\Gamma$ (the order parameter space). According to the classification of topological defects, the set of topological defects, defective on some subset $X\subseteq M$, is given by homotopy classes of maps $M\setminus X$ into the order parameter space, i.e.,
\begin{equation}
    \mathrm{hTop}\left[M\setminus X,\mathrm{Isom}\left(M\right)/\Gamma\right].
\end{equation}
Consider the image of the rotation group of a point in the space of lattices, i.e., $q\left(\mathrm{Rot}_M\left(p\right)\right)$, where
\begin{equation}
    q\colon\mathrm{Isom}\left(M\right)\to\mathrm{Isom}\left(M\right)/\Gamma.
\end{equation}
A single disclination at a point $p$ is then identified with the subset of defects which are defective on $\left\{p\right\}$, given as
\begin{equation}
    \mathrm{hTop}\left[M\setminus\left\{p\right\},q\left(\mathrm{Rot}\left(p\right)\right)\right]\hookrightarrow\mathrm{hTop}\left[M\setminus\left\{p\right\},\mathrm{Isom}\left(M\right)/\Gamma\right].
\end{equation}
This definition, when specialised to $\mathbb R^n$, reduces to the classical notion of disclination, a defect which has a  nontrivial structure only in the rotational symmetry group. 
\subsubsection{Multiple disclinations}
We consider the special case when $M$ is homogeneous, i.e., the isometries act transitively. In this case $\mathrm{Rot}\left(p\right)$ is independent of $p$. We call this group $\mathrm{Rot}$. We define a disclination, defective at a finite number of points $p_1,\dots,p_n$, to be a map
\begin{equation}
    f\colon M\setminus\left\{p_1,\cdots,p_n\right\}\to q\left(\mathrm{Rot}\right).
\end{equation}
Consider an embedding $b\colon B^n\to M$, where $B^n$ is $n$-dimensional ball, such that $p_1\in b\left(\mathrm{int}\;B^n\right)$ and $p_2,\cdots,p_n\not\in b\left(B^n\right)$ and $g\vert_{\partial B^n}\simeq c_{p_1}$ (boundary of the ball is contractible). Composing with $f$, we get a map $S^{n-1}\to q\left(\mathrm{Rot}\right)$, which represents the single point defect around $p_1$, among multiple disclinations represented by $f$. 
The condition of homogeneity is, however, very restrictive. For instance, all compact closed oriented surfaces with genus $>1$ are not homogeneous. An extension of our ideas to non-homogeneous spaces is presently unclear.
    
\subsection{A sheaf of tilings}
    One of the limitations pointed out in Section~\ref{sec:critique} arise from the incompatibility between the local and global nature of symmetries of lattices. A sheaf theoretic formalism could be an appropriate way to tackle this with suitable notions of unique gluing of the local data.
    Let $M$ and $S$ be two metric spaces. For any $U\subseteq M$ open, we define a tiling of $U$ by $S$ to be a collection of isometric embeddings $f_i\colon T_i\to U$ (where $S$ and $T_i$ are simply connected) such that
    \begin{itemize}
        \item $T_i\subseteq S$,
        \item $T_i$ is not necessarily all of $S$ if and only if $\mathrm{im}\;f_i\cap\partial U\neq\emptyset$, i.e., incomplete tiles can only occur at the boundaries,
        \item$\bigcup_i\mathrm{im}\;f_i=U$, i.e., the tiling covers $U$, and
        \item The interiors of the image of $f_i,f_j$ do not intersect unless $i=j$.
    \end{itemize}
    For any open set $U$, we define
    \begin{equation*}
        \mathcal F_{S,M}U=\text{set of tilings of }U\text{ by }S.
    \end{equation*}
    For $V\subseteq U$ and a tiling $f_i\colon T_i\to U$, introduce a tiling on $V$ by defining $T_i'=f_i^{-1}\left(\mathrm{im}\;f_i\cap V\right)$ and $f_i'\colon T_i'\to V$ by $f_i'=f_i\vert_{T_i'}$. We still need to verify whether this restriction map will make $\mathcal F$ a presheaf. Once verified, we believe that we can show that it is actually a sheaf (with a slightly modified definition). However, these are perfect tilings. If we can define a sheaf of deformed lattices, then we can understand defects as obstructions to extension of a tiling on some open subset (a section of this `sheaf of deformed tilings' on the open set) to a tiling on the entire space (a global section). This obstruction can be computed using some sort of `cohomology'. However, it is unclear at this stage, how one should define such a sheaf of deformed tilings, and further, if defined, as cohomology theories are usually studied for sheaves that take values in abelian categories (such as vector spaces), it requires further work to define the right sort of `cohomology' for the obstruction to extension to global sections of this set-valued sheaf of tilings.

\section*{\centering Acknowledgements}
AG acknowledges the financial support from SERB (DST) Grant No. MTR/2019/000578. We also acknowledge several insightful discussions with Mr. Animesh Pandey.

    \bibliographystyle{alpha}
    \bibliography{ref}
    \appendix
    \label{app:rn}
    
\section{Homotopy Theory}
The proofs for the following results can be found in \cite{May99,MP12}. 
\begin{theorem}
    \label{theo:pi0prescoprod}
    $\pi_0$ preserves coproducts.
\end{theorem}
\begin{theorem}
    \label{theo:wedgeiscoprod}
    Wedge is coproduct in $\mathrm{hTop}_*$.
\end{theorem}
\begin{theorem}
    \label{theo:relbtwbasdunbshtpy}
    If $Y$ is path-connected and $x\in X,y\in Y$ are well-pointed basepoints, then
    \begin{equation}
        \mathrm{hTop}\left[X,Y\right]\cong\mathrm{hTop}_*\left[\left(X,x\right),\left(Y,y\right)\right]/\pi_1\left(Y,y\right).
    \end{equation}
\end{theorem}
\begin{theorem}
    \label{theo;pinquotfreehom}
    If $X$ is path-connected and $y\in Y$ is a well-pointed basepoint, then
    \begin{equation}
        \mathrm{hTop}\left[S^n,Y\right]\cong\pi_n\left(Y,y\right)/\pi_1\left(Y,y\right).
    \end{equation}
\end{theorem}
\begin{theorem}
    \label{theo:pi1selfocnj}
    The action of $\pi_1$ on itself is by conjugation.
\end{theorem}
\begin{theorem}
    \label{theo:breakintopathcomp}
    For $X,Y$ locally path-connected, we have
    \begin{equation}
        \mathrm{hTop}\left[X,Y\right]=\prod_{A\in\pi_0\left(X\right)}\coprod_{B\in\pi_0\left(Y\right)}\mathrm{hTop}\left[A,B\right].
    \end{equation}
\end{theorem}
\begin{proof}
    We compute
    \begin{equation}
        \begin{split}
            \mathrm{hTop}\left[X,Y\right]\underset{1}{=}&\mathrm{hTop}\left[\coprod_{A\in\pi_0\left(X\right)}A,Y\right]\underset{2}{\cong}\prod_{A\in\pi_0\left(M\right)}\mathrm{hTop}\left[A,Y\right]\\
            \underset{3}{=}&\prod_{A\in\pi_0\left(X\right)}\mathrm{hTop}\left[A,\coprod_{B\in\pi_0\left(Y\right)}B\right]\underset{4}{\cong}\prod_{A\in\pi_0\left(X\right)}\coprod_{B\in\pi_0\left(Y\right)}\mathrm{hTop}\left[A,B\right],
        \end{split}
    \end{equation}
    where
    \begin{itemize}
        \item[$\underset{1}{=}$]$X$ is the coproduct of its path-connected components as it is locally path-connected,
        \item[$\underset{2}{\cong}$]Universal property of colimits,
        \item[$\underset{3}{=}$]$Y$ is the coproduct of its path-connected components as it is locally path-connected, and
        \item[$\underset{4}{\cong}$]The image of a connected space must be connected.
    \end{itemize}
\end{proof}
\begin{theorem}
    \label{theo:lesfunda}
    Let $G$ be a path-connected topological group and $H\leq G$. Then
    \begin{equation}
        \pi_n\left(H\right)\cong\pi_{n-1}\left(H\right)\cong0\implies\pi_n\left(G/H\right)\cong\pi_n\left(G\right).
    \end{equation}
\end{theorem}
\begin{theorem}
    \label{theo:univcovpi1}
    Let $G$ be a topological group and $q\colon\tilde G\to G$ be a universal covering map. Then
    \begin{equation}
        \pi_1\left(G/H\right)\cong q^{-1}\left(H\right).
    \end{equation}
\end{theorem}

\section{Topological Groups}
The proofs for the following results can be found in \cite{May99,MP12}, unless stated otherwise.
\begin{theorem}
    \label{theo:properorbstab}
    If $G$ acts properly on $X$ then $X/G$ is Hausdorff. In particular, each orbit is closed. The stabilizer of each point $x$ is compact and the map $G/\mathrm{Stab}_G\left(x\right)\to\mathrm{Orb}_G\left(x\right)$ is a homeomorphism. Moreover, if $G$ is Hausdorff then so is $X$.
\end{theorem}
\begin{proof}
    \cite{56490}
\end{proof}
\begin{theorem}
    \label{theo:lchcorbstab}
    Let $G$ be a locally compact, Hausdorff topological group with countable basis and $X$ a Hausdorff topological space with a continuous action of $G$ upon $X$. If the orbits are locally compact, then the map $G/\mathrm{Stab}_G\left(x\right)\to\mathrm{Orb}_G\left(x\right)$ is a homeomorphism for any $x\in X$.
\end{theorem}
\begin{proof}
    \cite{Gerrett10}
\end{proof}
\begin{theorem}
    \label{theo:connhomhomeo}
    If $X$ is a homogeneous $G$-space, then any two of its connected components are homeomorphic.
\end{theorem}
\begin{theorem}
    \label{theo:idquot}
    For a topological group $G$ and a subgroup $H\leq G$,
    \begin{equation}
        \left(G/H\right)_0\cong G_0/G_0\cap H.
    \end{equation}
\end{theorem}
\begin{theorem}
    \label{theo:connquot}
    For a topological group $G$ and a subgroup $H\leq G$
    \begin{equation}
        \pi_0\left(G/H\right)\simeq\pi_0\left(G\right)/p\left(H\right),
    \end{equation}
    where $p\colon G\to\pi_0\left(G\right)$ is the quotient map.
\end{theorem}
\begin{theorem}
    \label{theo:cartan}
    If $H\subset G$ is a closed subgroup of a (finite dimensional) Lie group, then $H$ is a sub-Lie group and a smooth sub-manifold such that its group operations are smooth functions with respect to the sub-manifold smooth structure.
\end{theorem}
\begin{proof}
    \cite{MSM_1952__42__1_0}
\end{proof}
\begin{theorem}
    \label{theo:covsemi}
    Let $G$ and $H$ be topological groups, $G$ be simply connected, and $H$ act on $G$ through an action $\phi\colon H\to\mathrm{Aut}\left(G\right)$. Let $p\colon\tilde H\to H$ be the universal cover of $H$. Then
    \begin{equation}
        \mathrm{id}_G\times p\colon G\rtimes_{\phi\circ p}\widetilde H\to G\rtimes_\phi H
    \end{equation}
    is also a universal cover.
\end{theorem}
\begin{proof}
    First we show that $\mathrm{id}_G\times p$ is a group homomorphism. For any $g_1,g_2\in G$ and $h_1,h_2\in\widetilde H$, we have
    \begin{equation}
        \begin{split}
            \left(g_1,h_1\right)\left(g_2,h_2\right)=\left(g_1\phi_{p\left(h_1\right)}\left(g_2\right),h_1h_2\right)\mapsto^{\mathrm{id}_G\times p}&\left(g_1\phi_{p\left(h_1\right)}\left(g_2\right),p\left(h_1h_2\right)\right)\\
            =&\left(g_1\phi_{p\left(h_1\right)}\left(g_2\right),p\left(h_1\right)p\left(h_2\right)\right)\\
            =&\left(g_1,p\left(h_1\right)\right)\left(g_2,p\left(h_2\right)\right).
        \end{split}
    \end{equation}
    Since topological properties do not depend on the fact that this is a semi-direct product and not an ordinary product (the underlying spaces are the same), we note that the domain is simply connected and the map is a covering map.
\end{proof}

\end{document}